\documentclass[12pt]{article}
\usepackage{bm, bbm, amsmath, amsthm, amssymb}
\usepackage{algorithm, algpseudocode}
\usepackage{enumitem}
\usepackage{natbib}
\usepackage{graphicx}
\usepackage{hyperref}
\usepackage{xcolor}
\hypersetup{
	colorlinks=true,
	linkcolor=black,
	citecolor=black,
	urlcolor=black
}
\usepackage{booktabs} %

\newtheorem{defn}{Definition}
\newtheorem{thm}{Theorem}

\newtheorem{prop}{Proposition}
\newtheorem{remark}{Remark}

\newcommand{\bvec}{\bm{\beta}}
\newcommand{\avec}{\bm{\alpha}}

\newcommand{\xvec}{\bm{x}}
\newcommand{\Xmat}{\bm{X}}
\newcommand{\Zmat}{\bm{Z}}
\newcommand{\Wmat}{\bm{W}}

\newcommand{\yvec}{\bm{y}}
\newcommand{\Ymat}{\bm{Y}}

\newcommand{\epvec}{\bm{\varepsilon}}
\newcommand{\epmat}{\bm{E}}

\newcommand{\reals}{\mathbbm{R}}
\newcommand{\Smat}{\bm{\Sigma}}

\newcommand{\Eset}{\mathcal{E}}
\newcommand{\envlp}{\Eset_{\Smat_{\xvec}}(\Smat_{\xvec \yvec})}
\newcommand{\hatenvlp}{\hat \Eset_{\Smat_{\xvec}}(\Smat_{\xvec \yvec})}
\newcommand{\envlpMB}{\Eset_{\bm M}( \bm B )}

\newcommand{\Gmat}{\bm \Gamma}

\newcommand{\evec}{\bm \eta}

\newcommand{\hatmat}[1]{ \widehat {\bm #1} }
\newcommand{\hatmatsub}[2]{\hatmat {#1}_{(#2)}}
\newcommand{\hatmatd}[1]{ \hatmatsub {#1}{d} }
\newcommand{\hatvec}[1]{ \hat {\bm #1} }
\newcommand{\hatvecsub}[2]{ \hat {\bm #1}_{(#2)} }

\newcommand{\matsub}[2]{{\bm #1}_{(#2)}}

\DeclareMathOperator{\E}{E}
\DeclareMathOperator{\Var}{Var}

\DeclareMathOperator{\Cov}{Cov}

\DeclareMathOperator{\spn}{span}
\DeclareMathOperator{\tr}{tr}
\DeclareMathOperator{\vect}{vec}
\DeclareMathOperator{\rank}{rank}

\DeclareMathOperator*{\argmin}{arg\,min}

\DeclareMathOperator{\diag}{diag}
\DeclareMathOperator{\new}{new}
\newcommand{\norm}[1]{\left\lVert#1\right\rVert}

\newcommand{\asto}{\overset{a.s.}{\to}}
\newcommand{\iid}{\overset{iid}{\sim}}

\graphicspath{{./figures}}

\newcommand{\blind}{1}
\newcommand{\withappendix}{1}

\def\spacingset#1{\renewcommand{\baselinestretch}%
	{#1}\small\normalsize} 
	
\spacingset{1}

\allowdisplaybreaks

\begin{document}
	
	\if1\withappendix{
	\date{}
	}\fi
	
	\if1\blind
	{
		\title{\bf Envelope-Guided Regularization for Improved Prediction in High-Dimensional Multivariate Regression}
		\author{Tate Jacobson%
				\hspace{.2cm}\\
			Department of Statistics, Oregon State University\\
			and \\
			Oh-Ran Kwon \\
			Department of Data Sciences and Operations,\\ University of Southern California}
		\maketitle
	} \fi
	
	\if0\blind
	{
		\bigskip
		\bigskip
		\bigskip
		\begin{center}
			{\LARGE\bf Envelope-Guided Regularization for Improved Prediction in High-Dimensional Multivariate Regression}
		\end{center}
		\medskip
	} \fi
	
	\bigskip

	\begin{abstract}
		Envelope methods perform dimension reduction of predictors or responses in multivariate regression,  exploiting the relationship between them to improve estimation efficiency. 
		While most research on envelopes has focused on their estimation properties, certain envelope estimators have been shown to excel at prediction in both low and high dimensions. 		
		In this paper, we propose to further improve prediction through envelope-guided regularization (EgReg), a novel method which uses envelope-derived information to guide shrinkage along the principal components (PCs) of the predictor matrix.
		We situate EgReg among other PC-based regression methods and envelope methods to motivate its development.
		We show that EgReg delivers lower prediction risk than a closely related non-shrinkage envelope estimator when the number of predictors $p$ and observations $n$ are fixed and in any alignment. 
		In an asymptotic regime where the true intrinsic dimension of the predictors and $n$ diverge proportionally, we find that the limiting prediction risk of the non-shrinkage envelope estimator exhibits a double descent phenomenon and is consistently larger than the limiting risk for EgReg. 
		We compare the prediction performance of EgReg with envelope methods and other PC-based prediction methods in simulations and a real data example, observing improved prediction performance over these alternative approaches in general.
	\end{abstract}
	
	\noindent%
	\vfill
	
	\newpage
	\spacingset{1.9}
	
	\section{Introduction}
	Envelope methods aim to improve estimation efficiency in regression by identifying and discarding the components of the variables which are immaterial to the modeling goal.
	Envelopes were first introduced for dimension reduction in the space of the responses in multivariate linear regression \citep{Cook2010}, but have since been adapted for predictor reduction \citep{Cook2013} and simultaneous predictor and response reduction \citep{Cook_simultaneous_2015}. 
	Envelope methodology has been extended beyond linear models, with \citet{Cook_foundations_2015} introducing a general envelope framework and applying it to generalized linear models and Cox regression, and more recent works developing envelope models for reduced rank regression \citep{cook_reduced_rank_2015}, Huber regression \citep{zhou_enveloped_2024} and quantile regression \citep{Ding2021}.
	
	Most research on envelope methods has focused on their ability to improve estimation efficiency rather than prediction.
	One notable exception is partial least squares (PLS) regression \citep{wold_multivariate_1983}, a core method from chemometrics which \citet{Cook2013} showed to be estimating an envelope in the predictor space---known as a \textit{predictor envelope}.
	PLS owes its enduring popularity to its prediction performance in high-dimensional applications.
	The connection between PLS and envelopes has facilitated the theoretical study of PLS predictions:
	\citet{cook_big_2018} proved the consistency of single-component PLS predictions and \citet{cook_partial_2019} derived the rate of multiple-component PLS predictions under various scenarios, finding that they perform best in high-dimensional regressions where many predictors are related to the response.
	
	While research on envelope-based prediction methods has largely focused on PLS, \citet{Cook2013} demonstrated that a likelihood-based envelope estimator can deliver more accurate predictions than PLS in practice, although their analysis was limited to low dimensions. 
	Separately, \citet{rimal2019comparison} used extensive simulations and real data analyses to show that other predictor envelope estimators provide stable and improved performance compared to PLS in high dimensions. 
	These empirical findings suggest that PLS is merely one of many possible methods for improving prediction via envelopes.

	In this paper, we develop a new way of exploiting the envelope structure to improve prediction in high-dimensional regressions.
	Drawing on the connection between predictor envelopes and the principal components (PCs) of the predictor matrix \citep{Zhang2023}, we propose a novel envelope-based shrinkage method which we call \textit{envelope-guided regularization} (EgReg).
	Unlike the best-known PC-based variance reduction methods---principal components regression (PCR) and ridge regression---EgReg utilizes information from the joint covariance structure of the predictors and the responses to supervise shrinkage within the predictor space.
	In taking this approach, the EgReg estimator reduces the variance while avoiding the bias that comes from discarding or shrinking PCs which may be material to the estimation of the response. 
	EgReg differs from other envelope methods in a few respects:
	The EgReg estimator is, to our knowledge, the first to use an envelope-based criterion to guide shrinkage, as well as the first shrinkage method developed within the context of predictor envelopes.
	Whereas most envelope estimators can only be computed when there are fewer predictors $p$ than observations $n$, the EgReg estimator can be computed when $p$ and $n$ are in any alignment.
	Lastly, whereas the theoretical guarantees for other high-dimensional envelope methods rely on sparsity assumptions \citep{su_sparse_2016, Zhang2023}, the guarantees for EgReg hold in both sparse and non-sparse regressions.
	
	Theoretically, we demonstrate that for any $n$ and $p$ the EgReg estimator can always attain a smaller upper bound on its prediction risk than the Non-Iterative Envelope Component Estimation (NIECE) estimator of \citet{Zhang2023}, a closely related but non-shrinkage envelope method.
	In addition, we derive the limiting prediction risk for the EgReg and NIECE estimators as the true intrinsic dimension of the predictor envelope $u^*$ and $n$ diverge, with their ratio $u^*/n$ converging to a positive constant---a first, to our knowledge, for any predictor envelope estimator. 
	Notably, we identify that the NIECE estimator exhibits a double descent asymptotic prediction risk curve---the risk quickly diverges to $\infty$ as the limiting ratio approaches a certain threshold, then quickly decreases once it exceeds that threshold. 
	
	We observe that the EgReg estimator is particularly useful when a predictor envelope structure exists and the envelope does not fall in a space generated solely by high-variance PCs of the predictor matrix.
	However, even in the absence of an envelope structure the EgReg estimator is not at a disadvantage, as we find that its prediction performance remains comparable to that of other multivariate regression methods. 

	The rest of the paper is organized as follows:
	In Section \ref{sec:motivation}, we lay out the motivation for EgReg by comparing three existing PC-based regression methods: PCR, ridge regression, and NIECE. 
	In Section \ref{sec:egreg}, we introduce the EgReg estimator and situate it among these methods.
	In Section \ref{sec:theory}, we compare the prediction risks of EgReg and NIECE in finite samples and as the number of samples and the true envelope dimension diverge.
	In Section \ref{sec:sims}, we conduct simulations comparing EgReg with alternative multivariate regression methods.
	In Section \ref{sec:real data}, we compare the methods' prediction performance on real near-infrared spectroscopy data.
	We close with some discussion in Section \ref{sec:discussion}.
	Technical proofs can be found in the
	\if1\withappendix{appendix.
	}\fi \if0\withappendix{supplementary material.}\fi
	
	\section{PCR, Ridge, and Predictor Envelopes} \label{sec:motivation}
	We start with a multivariate linear regression model
	\begin{equation}
		\yvec' = \xvec'\bvec^* + \epvec', \label{eqn:linear model}
	\end{equation}
	where $\yvec \in \reals^q$, $\xvec \in \reals^{p}$, $\bvec^* \in \reals^{p \times q}$, and $\epvec \in \reals^q$.
	Suppose that $\xvec \sim P_{\xvec}$, $\epvec \sim P_{\epvec}$, and $\xvec$ and $\epvec$ are independent, with $\E[\xvec] = \bm 0$, $\Var(\xvec) = \Smat_{\xvec}$, $\E[\epvec] = \bm 0$, and $\Var(\epvec) = \Smat_{\epvec}$. 
	We omit an intercept without loss of generality for streamlined communication.
	
	Suppose we observe $n$ independent and identically distributed (i.i.d.) copies of $(\yvec, \xvec)$: $\{(\yvec_i, \xvec_i)\}_{i = 1}^n$, where $n$ and $p$ can be in any alignment and $q$ is fixed. 
	Define $\Xmat = (\xvec_1, \ldots, \xvec_n)'$, $\Ymat = (\yvec_1, \ldots, \yvec_n)'$, and $\epmat = (\epvec_1, \ldots, \epvec_n)'$.
	Then we can write the model in matrix form as
	\begin{equation}
		\Ymat = \Xmat\bvec^* + \epmat \label{eqn:matrix linear model}
	\end{equation}
	where $\E[ \epmat ] = \bm 0$ and $\Var(\epmat) = \Smat_{\epvec} \otimes \bm I_n$, with `$\otimes$' denoting the Kronecker product.
	
	We introduce some necessary notation. %
	Let $r = \rank(\Xmat) = \min\{n,p\}$.
	We take the singular value decomposition (SVD)
	of $\Xmat$: $\Xmat = \hatmat U \hatmat D \hatmat V'$, where $\hatmat{U} = [\hatvec u_1 \cdots \hatvec u_r] \in \reals^{n \times r}$ and $\hatmat{V}  = [\hatvec v_1 \cdots \hatvec v_r] \in \reals^{p \times r}$ are the left and right singular vectors, respectively, and $\hatmat D = \diag\{\hat \sigma_j\}$ is the diagonal matrix of singular values $\hat \sigma_1 \geq \hat \sigma_2 \geq \ldots \geq \hat \sigma_r > 0$ of $\Xmat$.
	By definition $\hatmat{V}'\hatmat{V} = \hatmat{U}'\hatmat{U} = \bm I_r$.
	We call
	$\Xmat \hatvec v_i = \hatvec u_i \hat \sigma_i$ 
	the $i$th PC of $\Xmat$ and $\hatvec u_i$ the $i$th normalized PC. 
	Note that the squared singular value $\hat \sigma_i^2$ corresponds to the variance of the $i$th PC.
	We use $\bm I_p$ to denote the $p \times p$ identity matrix and drop the subscript when the dimension is clear from the context.
	We use $\spn(\bm A)$ to denote the space spanned by the columns of a matrix $\bm A$.

	\subsection{Thresholding versus shrinkage}
	To situate EgReg among existing methods, we review how the popular PCR and ridge estimators apply thresholding and shrinkage, respectively, within the space of the PCs of $\Xmat$ to reduce the variance of the models' predictions.

	In PCR, one regresses $\Ymat$ on the first $d$ PCs of $\Xmat$, where $d < r$ is a user-specified tuning parameter.
	We can express the reduced predictor matrix consisting of the first $d$ PCs as $\hatmat{U}_d \hatmat{D}_d$, where $\hatmat{U}_d$ and $\hatmat{D}_d$ denote the first $d$ columns of $\hatmat U$ and $\hatmat D$.
	Using this notation, the PCR estimator of $\Ymat$ can be written as
	\begin{equation}
		\widehat \Ymat^{\text{PCR} }
		= \hatmat{U}_d \hatmat{U}_d' \Ymat = \sum_{j = 1}^d \hatvec u_j \hatvec u_j'\Ymat. \label{eqn:PCR}
	\end{equation}
	In this expression, we see that the PCR estimator projects $\Ymat$ onto the space generated the first $d$ PCs,
	effectively discarding the $r-d$ lowest-variance PCs from the model.

	We can similarly express the ridge estimator in terms of the PCs, as
	\begin{equation}
		\widehat \Ymat^{\text{ridge} } 
		= \hatmat U \hatmat D(\hatmat D^2 + \lambda \bm I_p)^{-1}\hatmat D \hatmat U' \Ymat 
		= \sum_{j = 1}^r \hatvec u_j \frac{\hat \sigma_j^2}{\hat \sigma_j^2 + \lambda} \hatvec u_j'\Ymat \text{.} \label{eqn:ridge}
	\end{equation}
	Here we see that the ridge estimator computes the coordinates of $\Ymat$ with respect to the basis of the normalized PCs, $\{\hatvec u_1, \ldots, \hatvec u_r\}$, then shrinks these coordinates by $ \hat \sigma_j^2/(\hat \sigma_j^2 + \lambda) $.
	Examining this scaling term, we see that PCs with lower variance $\hat \sigma_j^2$ are subject to greater shrinkage, reducing their influence on the estimate of $\Ymat$.
	Unlike in PCR, however, all $r$ PCs are given some weight in the ridge estimator.
	
	Both PCR and ridge modulate the effects of the PCs in the prediction based on their singular values.
	The goal for both methods is to reduce the variance without introducing too much bias.
	The key difference between them is that PCR performs \textit{thresholding} in the space of the PCs while ridge applies \textit{shrinkage}.
	
	\subsection{Envelope scores versus singular values}
	
	One commonly cited drawback of PCR is that the highest-variance PCs of $\Xmat$ may not be those most material to the estimation of $\Ymat$ (see, for example, \citet{jolliffe_note_1982} and \citet{lang_simple_2020}, who called this practice the ``blind selection'' of PCs).
	
	We can demonstrate the blind selection problem with a toy example:
	Suppose that $r=10$, that the singular values of $\Xmat$ are $\hatvec{\sigma} = (1, \ldots, 1, 1/1000)'$, and that the true relationship between $\Xmat$ and $\Ymat$ is 
	$\Ymat = 1000\Xmat \hatvec{v}_{10} + \epmat$---that is, only the 10th PC of $\Xmat$ is material to the regression.
	Because PCR sorts the PCs based on their singular values, the best it can do in this scenario is to select all ten PCs of $\Xmat$, accumulating variance from the nine noise PCs in the process, and the worst it can do is to fail to select the tenth PC at all, increasing the bias.
	Ridge regression suffers from a similar issue, though to a lesser extent: in this toy scenario, it would apply more shrinkage along $\hatvec{u}_{10}$ than the noise PCs, penalizing exactly the wrong components of $\Xmat$ and increasing the bias without a commensurate reduction in variance.
	We term this practice ``blind shrinkage'', mirroring the blind selection performed by PCR.
	
	For both PCR and ridge there is a potential for misalignment between the metric used to select and shrink the PCs---the singular values, which depend only on the covariance structure of $\Xmat$---and the objective of reducing the variance in estimating $\Ymat$.
	As we will see, this misalignment can be avoided by instead using an envelope-based criterion.
	
	\subsubsection{Review of envelope scores}
	\citet{Zhang2023} drew a connection between envelopes and PCR by introducing a new means of selecting PCs to use in the regression: the \textit{envelope scores}.
	Before introducing the scores, we must first review envelopes.
	For a subspace $\mathcal{B} \subseteq \reals^p$, we define $\bm P_{\mathcal B}$ to be the projection onto $\mathcal B$ and $\bm Q_{\mathcal B} = \bm I_p - \bm P_{\mathcal B} $ to be the projection onto the orthogonal complement of $\mathcal B$.
	Using this notation, we can define a few key terms:
	\begin{defn}
		A \emph{reducing subspace} of a symmetric matrix $\bm M \in \reals^{p \times p}$ is defined as a subspace $\mathcal{R} \subseteq \reals^p$ such that
		$
		\bm M = \bm P_{\mathcal R} \bm M \bm P_{\mathcal R} + \bm Q_{\mathcal R} \bm M \bm Q_{\mathcal R}.
		$
		An \emph{$\bm M$-envelope} of a subspace $\mathcal B = \spn (\bm B) \subseteq \reals^p$, denoted as $\envlpMB$, is the intersection of all reducing subspaces of $\bm M$ that contain $\mathcal B$.
	\end{defn}
	We assume without loss of generality that $\bm B \in \reals^{p\times  p}$ is symmetric and positive semi-definite
	(since the column space of $\bm B$ is the same as the column space of $\bm B \bm B'$, we can always replace $\bm B$ with $\bm B \bm B'$ to meet these conditions).
	
	\textit{Predictor envelopes} perform dimension reduction in the space of the predictors of the linear regression model \eqref{eqn:linear model}.
	Recall that $\Smat_{\xvec} = \Var(\xvec)$ and define $\Smat_{\xvec \yvec} = \Cov(\xvec, \yvec)$.
	Within the general envelope framework, the predictor envelope is the $\Smat_{\xvec}$-envelope of $\spn(\Smat_{\yvec \xvec} )$---i.e. $\envlp$.
	\citet{Cook2013} showed that $\mathcal{R}$ is a reducing subspace of $\Smat_{\xvec}$ if and only if $\bm P_{\mathcal R}\xvec$ is uncorrelated with $\bm Q_{\mathcal R}\xvec$ and that $\spn(\Smat_{\yvec \xvec} ) \subseteq \mathcal R$ if and only if $\yvec$ is uncorrelated with $\bm Q_{\mathcal R}\xvec$ given $\bm P_{\mathcal R}\xvec$.
	We refer to $\bm P_{\mathcal R}\xvec$ and $\bm Q_{\mathcal R}\xvec$ as the \textit{material} and \textit{immaterial} parts of $\xvec$ for the regression, respectively.
	Viewed in these terms, $\envlp$ is the smallest subspace $\mathcal R$ such that the projection of $\xvec$ onto $\mathcal R$ contains all of the information about $\yvec$ that can be found in $\xvec$.
	
	Suppose that $\Gmat$ is an orthogonal basis for $\envlp$.
	Then $\bvec^* \in \spn(\Gmat)$, so there exists $\bm \alpha^*$ such that $\bvec^* = \Gmat \bm \alpha^*$ and \eqref{eqn:linear model} can be expressed as
	$
		\yvec' = \xvec' \Gmat \bm \alpha^* + \epvec
	$.
	If $\Gmat$ was known, then we could estimate $\bm \alpha^*$ using the ordinary least squares estimator $\hatvec{\alpha} = (\Gmat' \Xmat'\Xmat \Gmat)^{-1}\Gmat'\Xmat'\Ymat$ and estimate $\bvec^*$ with $\hat \bvec = \Gmat \hatvec{\alpha}$.
	\citet{Cook2013} show that using this envelope estimate of $\bvec^*$ can lead to a substantial reduction in the prediction variance.
	
	Define  $u^* = \dim \{\envlpMB \}$.
	\citet{Zhang2023} proposed the Non-Iterative Envelope Component Estimation (NIECE) algorithm (Algorithm \ref{algo:NIECE population}) to obtain $\envlpMB$ using the eigenvectors of $\bm M$ when $\bm M$, $\bm B$, and $u^*$ are known, referring to this as the ``population-level'' NIECE algorithm. 

	\begin{algorithm}[h]
		\caption{NIECE algorithm \citep{Zhang2023}} \label{algo:NIECE population}
		\begin{algorithmic}[1]
			\State \textbf{Input:} Symmetric $p \times p$ matrices $\bm M >0$ and $\bm B \geq 0$, a number of eigenvectors of $\bm M$ to use $d$, and the envelope dimension $u^*$. Note that $0\leq u^* \leq d \leq p$.
			\State Obtain the first $d$ orthogonal eigenvectors of $\bm M$: $\bm V_d = [\bm v_1, \ldots, \bm v_d] \in \reals^{p \times d}$.%
			\State Calculate the \textbf{envelope scores} $\phi_j = \bm v_j' \bm B \bm v_j$.
			\Statex Sort them in descending order $\phi_{(1)} \geq \ldots \geq \phi_{(d)} $.
			\Statex Define $\bm v_{(j)}$ such that $\phi_{(j)} = \bm v_{(j)}' \bm B \bm v_{(j)}$.
			\State\textbf{Output:} the envelope $\envlpMB = \spn(\bm v_{(1)}, \ldots, \bm v_{(u^*)})$
		\end{algorithmic}
	\end{algorithm}

	It is straightforward to show that the population-level NIECE algorithm recovers the true envelope under mild conditions:
	\begin{prop}[\citet{Zhang2023}]\label{prop:NIECE envelope recovery}
		If $\envlpMB \subseteq \spn (\bm V_d)$ for some $d \leq p$ and the eigenvalues of $\bm M$ are distinct, then 
		\begin{equation}
			\envlpMB = \sum_{\substack{i = 1\\ \bm v_i'\bm B \bm v_i \neq 0}}^d \spn( \bm v_i ), \label{eqn:Zhang envlp expression}
		\end{equation}
		the envelope can be expressed as $\envlpMB = \spn(\bm v_{(1)}, \ldots, \bm v_{(u^*)})$, and the envelope scores satisfy $  \phi_{(1)}	\geq \ldots \geq \phi_{(u^*)} > \phi_{(u^*+1)} = \ldots = \phi_{(d)} = 0.$
	\end{prop} 
	
	\begin{remark}
		Note that the condition that $\envlpMB \subseteq \spn (\bm V_d)$ is trivially true for $d = p$.
	\end{remark}
	
	Proposition \ref{prop:NIECE envelope recovery} establishes that the ranking of the eigenvectors $\bm v_j$ based on their envelope scores enables the NIECE algorithm to recover $\envlpMB$.
	Motivated by this key property of the population-level algorithm, \citet{Zhang2023} proposed a plug-in estimator for NIECE, replacing $\bm M$ and $\bm B$ by their sample analogs:
	In the case of predictor envelopes $\bm M = \Smat_{\xvec}$ and $\bm B = \Smat_{\xvec \yvec} \Smat_{\yvec \xvec}$, so the sample versions are $\hatmat M = \bm S_{\xvec} =\frac{1}{n} \Xmat'\Xmat$ and $\hatmat {B} = \bm S_{\xvec \yvec} \bm S_{\yvec \xvec}$, where $\bm S_{\xvec \yvec} = \frac{1}{n} \Xmat'\Ymat$.
	For the sample-level NIECE algorithm, the user specifies the number of PCs $u \leq d$ to use.
	The sample envelope scores are $\hat \phi_j = \hatvec v_j' \bm S_{\xvec \yvec} \bm S_{\yvec \xvec } \hatvec v_j$, where $\hatvec v_1, \ldots, \hatvec v_d$ are the first $d$ orthogonal eigenvectors of $\bm S_{\xvec}$, which we obtain via the SVD of $\Xmat$ (\citet{Zhang2023} also propose a ``sparse'' version based on the penalized matrix decomposition of \citet{Witten2009}, though we will focus on the SVD version here).
	We rank the sample envelope scores $\hat \phi_{(1)} \geq \ldots \geq \hat \phi_{(d)}$ and define the $\hatvecsub{v}{j}$ such that $\hat \phi_{(j)} =  \hatvecsub{v}{j}'\bm S_{\xvec \yvec} \bm S_{\yvec \xvec} \hatvecsub{v}{j}$.
	Using this ranking of the eigenvectors, the NIECE envelope estimator is $\hatenvlp = \spn(\hatvec v_{(1)}, \ldots, \hatvec v_{(u)} )$.
	
	\subsubsection{NIECE as a thresholding method}
	As with PCR and ridge, we can express the NIECE estimator of $\Ymat$ in terms of the PCs of $\Xmat$.
	Recall that $\Xmat = \hatmat U \hatmat D \hatmat V'$.
	Define $\hatmatsub{V}{u} = [\hatvecsub{v}{1}, \ldots, \hatvecsub{v}{u} ]$, $\hatmatsub{U}{u} = [\hatvecsub{u}{1}, \ldots, \hatvecsub{u}{u}] $, and  $\hatmatsub{D}{u} = \diag\{ \hat \sigma^2_{(1)} , \ldots, \hat \sigma^2_{(u)} \}  $ where the columns are ranked based on their sample envelope scores.
	The NIECE-estimated basis for $\hatenvlp$ is $\hatmatsub{V}{u}$.
	As such, the reduced predictor matrix is
	$\Xmat \hatmatsub{V}{u} = \hatmatsub{U}{u} \hatmatsub{D}{u}$ and the NIECE estimate of $\Ymat$ is
	\begin{equation}
		\widehat \Ymat^{\text{NIECE}} 
		=\hatmatsub{U}{u} \hatmatsub{U}{u}' \Ymat = \sum_{j = 1}^u \hatvecsub{u}{j} \hatvecsub{u}{j}'\Ymat. \label{eqn:NIECE}
	\end{equation}
	
	Comparing \eqref{eqn:PCR} and \eqref{eqn:NIECE}, we can draw clear parallels between PCR and NIECE:
	Both are thresholding methods, discarding some of the PCs of $\Xmat$ to perform dimension reduction.
	The key difference between them is that PCR ranks and discards PCs based on their variances $\hat \sigma_j^2$ whereas NIECE ranks and discards PCs based on their envelope scores $\hat \phi_j$.
	(PLS, an iterative predictor envelope estimator, can also be thought of as a thresholding method, though we will not dwell on it here.)
	By writing the envelope scores as $\hat \phi _j = \frac{1}{n^2} \hatvec{v}_j'\Xmat'\Ymat \Ymat'\Xmat \hatvec{v}_j$,
	we can see that NIECE selects the $u$ PCs that are most strongly correlated with $\Ymat$.
	In this way NIECE uses information from the response to guide the selection of the PCs, thereby avoiding PCR's practice of blind selection.

	\section{Envelope-Guided Regularization} \label{sec:egreg}
	As with other envelope methods, the goal of NIECE is to estimate the envelope subspace $\envlp$ and thereby improve estimation efficiency.
	While a thresholding approach is well-suited to this task, a different approach may be better for prediction. 
	Prior research has shown that ridge regression tends to provide more accurate predictions than PCR and PLS in practice \citep{frank_statistical_1993}, vindicating its shrinkage-based approach over their thresholding approaches.
	In this same study, PLS was shown to outperform PCR, providing empirical support for envelope-based prediction methods.
	
	Inspired by these empirical results, we propose envelope-guided regularization (EgReg), a shrinkage estimator based on the envelope scores.
	The basic motivation for EgReg is to obtain a ``greedy'' estimate of the envelope basis using the PCs of $\Xmat$ and then apply shrinkage along the PCs using information from the envelope structure to improve prediction.
	Algorithm \ref{algo:EgReg} details how we compute the EgReg estimator.

	\begin{algorithm}[h]
		\caption{Envelope-Guided Regularization (EgReg) Algorithm} \label{algo:EgReg}
		\begin{algorithmic}[1]
			\State \textbf{Input:} $d \leq r$, $\lambda$
			\State Compute the SVD of $\Xmat$: $\Xmat = \hatmat U \hatmat D \hatmat V'$.
			\State Calculate the envelope scores for the first $d$ right singular vectors $\hat \phi_j = \hat{\bm v}_j' \bm S_{\xvec \yvec} \bm S_{\yvec \xvec } \hat{\bm v}_j$ and sort them in decreasing order $\hat \phi_{(1)} \geq \ldots \geq \hat \phi_{(d)} \geq 0$. 
			\Statex Let $\hatmat \Phi= \diag\{ \hat \phi_j \}_{j=1}^p$ and $\hatmatd{\Phi} = \diag\{ \hat \phi_{(j)} \}_{j=1}^d$.
			\Statex Define $ \hatvecsub{v}{j} $ such that $\hat \phi_{(j)} = \hatvecsub{v}{j}' \bm S_{\xvec \yvec} \bm S_{\yvec \xvec } \hatvecsub{v}{j}$ for $j = 1,\ldots, d$. 
			\State Compute the predictor reduction matrix $ \hatmat \Gamma = \hatmatd{V} \hatmatd{D}^{-1}\hatmatd\Phi^{1/2} \in \reals^{p \times \hat d}$.
			\State Solve the ridge problem with respect to the reduced predictors $\Xmat \hatmat \Gamma$:
			\begin{equation}
				\hat \evec = \argmin_{\eta \in \reals^d } \norm{\Ymat - \Xmat \hatmat \Gamma \evec}_F^2 + \lambda \norm{\evec}_F^2 \label{eqn:reduced ridge problem}
			\end{equation}
			\State Compute
			$
			\hat \bvec = \hatmat \Gamma \hat \evec
			$
		\end{algorithmic}
	\end{algorithm}

	Because $\Xmat \hatmat \Gamma = \hatmatd{U} \hatmatd{\Phi}^{1/2} $, the EgReg estimator of $\Ymat$ can be written as
	\begin{equation}
		\widehat \Ymat^{\text{EgReg}}
		= \hatmatd{U}  \hatmatd{\Phi}^{1/2} ( \hatmatd{\Phi} + \lambda \bm I)^{-1}  \hatmatd{\Phi}^{1/2} \hatmatd{U}' \Ymat		
		= \sum_{j = 1}^d \hatvecsub{u}{j} \frac{\hat \phi_{(j)}}{\hat \phi_{(j)} + \lambda} \hatvecsub{u}{j}' \Ymat \label{eqn:envelope guided ridge}
	\end{equation}
	In \eqref{eqn:envelope guided ridge}, we see that the EgReg estimator applies shrinkage along the normalized PCs of $\Xmat$, with greater shrinkage applied to PCs with lower envelope scores.
	The shrinkage mechanism is similar to that of the ridge estimator \eqref{eqn:ridge}, but with a different criterion driving the level of shrinkage applied to each PC---we can describe the regularization in \eqref{eqn:envelope guided ridge} as ``envelope-guided'' and the regularization in \eqref{eqn:ridge} as ``singular-value-guided''.

	Envelope models typically assume that the true envelope dimension $u^*$ is small relative to $n$ and $p$.
	As such, the user-selected dimension $u$ for NIECE is generally also small.
	Because the motivation for EgReg is not to recover $\envlp$ but rather to improve prediction, the number of PCs $d$ used for EgReg does not need to be close to $u^*$. 
	It may even be best to use all $r$ PCs of $\Xmat$ if doing so reduces the prediction variance without introducing too much bias.
	(In the empirical studies in sections \ref{sec:sims} and \ref{sec:real data}, we will see that using all $r$ PCs often leads to the best prediction performance.)
	
	In our review of PCR, ridge, and NIECE, we introduced two axes along which these methods can be divided: thresholding versus shrinkage methods and envelope-score-based versus singular-value-based methods.
	Considered in these terms, EgReg can be viewed as a shrinkage-based alternative to NIECE or as an envelope-score-based analog to ridge regression.
	Table \ref{tbl:PC-based methods} situates EgReg among these other PC-based prediction methods.
	
	\begin{table}[h]
		\caption{Taxonomy of PC-based methods}
		\label{tbl:PC-based methods}
		\centering
		\begin{tabular}{l|cc|}
			\multicolumn{1}{l}{} & Thresholding & \multicolumn{1}{c}{Shrinkage} \\  
			\cline{2-3}
			Singular values & PCR & Ridge  \\
			Envelope scores & Envelopes (NIECE) & \textbf{EgReg}\\
			\cline{2-3} 
		\end{tabular}
	
	\end{table}
	
	As we have noted, the envelope scores measure the strength of the association between the PCs of $\Xmat$ and $\Ymat$.
	By using the envelope scores rather than the singular values to guide PC shrinkage and selection, EgReg and NIECE avoid the PCR and ridge practices of blind selection and blind shrinkage.
	However, as a thresholding procedure NIECE still runs the risk of discarding potentially important PCs.
	EgReg mitigates this risk with its shrinkage approach.
	For the remainder of the paper, we examine the theoretical and empirical prediction performance of EgReg relative to alternative PC-based methods.
	
	\section{Theory}\label{sec:theory}
	To support our characterization of EgReg as a method of improving on the predictions of other envelope estimators, we compare the theoretical prediction risks of EgReg and NIECE. 
	We first examine the case where $n$ and $p$ are fixed, allowing them to be in any alignment.
	We then examine the limiting prediction risks of EgReg and NIECE as $n$ and the true envelope dimension $u^*$ diverge proportionally, with $u^*/n \to \gamma \in (0, \infty)$ .
	
	\subsection{Prediction risk}
	Suppose that $\xvec_{\new} \sim P_{\xvec}$ is independent of the training sample. 
	For an estimator $\hat \bvec$, we define the prediction risk to be:
	\begin{equation}
		R(\hat \bvec|\Xmat) = \E\left[ \norm{ \xvec_{\new}' ( \hat \bvec - \bvec^* ) }_2^2 \middle| \Xmat \right]. 
	\end{equation}
	
	We denote the NIECE estimator with tuning parameter $u$ by $\hat \bvec^N(u)$ and the EgReg estimator with tuning parameters $d$ and $\lambda$ by $\hat \bvec^E(d,\lambda)$.
	The EgReg estimator can be written as
	$
		\hat \bvec^E(d,\lambda) 
		= \hatmatd{V} \hatmatd{D}^{-1} \hatmatd \Phi (\hatmatd \Phi + \lambda \bm I_d )^{-1} \hatmatd{U}'\Ymat 
	$
	and the NIECE estimator can be written as
	$
		\hat \bvec^N
		= \hatmat{V}_{(u)} ( \hatmat{V}_{(u)}' \Xmat' \Xmat \hatmat{V}_{(u)} )^{-1} \hatmat{V}_{(u)}' \Xmat'\Ymat = \hatmat{V}_{(u)} \hatmat{D}_{(u)}^{-1} \hatmat{U}_{(u)}' \Ymat.
	$
	
	Suppose, without loss of generality, that $u \leq d$.
	Then both the EgReg and NIECE estimators fall within $\spn(\hatmatd{V})$.
	Considering the projection of $\bvec^*$ onto $\spn(\hatmatd V)$, we see that the prediction error of both estimators can be broken down as
	$
		\xvec_{\new}' ( \hat \bvec - \bvec^* )
		= \xvec_{\new}' ( \hat \bvec - \bm P_{\hatmatd{V}}\bvec^* ) + \xvec_{\new}' ( \bm P_{\hatmatd{V}}\bvec^* - \bvec^*) 
	$ 
	and that the prediction risk admits the following upper bound in terms of reducible and irreducible components
	\begin{align}
		R(\hat \bvec|\Xmat) 
		& \leq 2\underbrace{\E\left[ \norm{ \xvec_{\new}' ( \hat \bvec - \bm P_{\hatmatd{V}}\bvec^* ) }_2^2 \middle| \Xmat \right]}_{\text{reducible risk}}
		+ 2\underbrace{\E\left[ \norm{ \xvec_{\new}' ( \bm P_{\hatmatd{V}} \bvec^* - \bvec^* ) }_2^2 \middle| \Xmat \right]}_{\text{irreducible risk}} \label{eqn:reducible irreducible risk decomp}%
	\end{align}
	The irreducible risk captures the bias induced by projecting $\bvec^*$ into the column space of $\hatmatd{V}$.
	Because both estimators fall within $\spn(\hatmatd{V})$, they induce the same irreducible risk.
	As such, we will compare these estimators based on their reducible risks.
	
	It is straightforward to show that the reducible prediction risk admits the following bias-variance decomposition:
	\begin{align}
		R_r(\hat \bvec|\Xmat) 
		& = \norm{ \vect ( \E[\hat \bvec | \Xmat] - \bm P_{\hatmatd{V}} \bvec^*) }^2_{\bm I_q \otimes \Smat_{\xvec}} 
		+ \tr\{\Var(\vect (\hat \bvec) | \Xmat)(\bm I_q \otimes \Smat_{\xvec}) \} \notag \\
		& =  B_r(\hat \bvec) + V(\hat \bvec)  \label{eqn:generic bv decomp}
	\end{align}
	where $\norm{\bm v}_{\bm W}^2 = \bm v' \bm W \bm v$ is the $\bm W$-inner product for a vector $\bm v \in \reals^q$ and matrix $\bm W \in \reals^{q\times q}$.
	We refer to $B_r(\hat \bvec)$ as the reducible bias component and $ V(\hat \bvec) $ as the variance component.

	\subsection{Finite-sample risk comparison}
	To facilitate our theoretical analysis,  we assume that the envelope scores $\{\hat \phi_j\}_{j = 1}^r$ are given and independent of the response $\Ymat$---this would occur, for example, if we had two samples $(\Xmat, \Ymat)$ and $(\Xmat, \widetilde{\Ymat})$ with the same predictors and independent realizations of the response, and estimated the envelope scores using only $(\Xmat, \widetilde{\Ymat})$.
	We treat $\hatmat \Phi$ as given and study the conditional risk $R_r(\hat \bvec|\Xmat,\hatmat \Phi)$. 
	Consequently, the risk does not account for the variability arising from the estimation of these scores. 
	
	\subsubsection{Bias-variance decompositions}
	We use the bias-variance decomposition \eqref{eqn:generic bv decomp} to compare the reducible prediction risks for EgReg and NIECE.
	One can show that the EgReg estimator has reducible bias
	\begin{align}
		\E\left[ \hat \bvec^E(d, \lambda) \middle | \Xmat, \hatmat \Phi \right] - \bm P_{\hatmatd{V}}\bvec^*
		& = \hatmatd{V}  \left[  \hatmatd \Phi (\hatmatd \Phi + \lambda \bm I_d )^{-1} - \bm I_d  \right] \hatmatd{V}'\bvec^* \notag \\
		& = - \lambda \hatmatd{V}(\hatmatd \Phi + \lambda \bm I_d )^{-1} \hatmatd{V}'\bvec^* \label{eqn:bias for egreg}
	\end{align}
	and variance
	\begin{align*}
		\Var( \vect(\hat \bvec^E(d, \lambda) ) | \Xmat, \hatmat \Phi)
		& = \Var( \vect(  \hatmatd{V} \hatmatd{D}^{-2} \hatmatd \Phi (\hatmatd \Phi + \lambda \bm I_d )^{-1} \hatmatd{V}' \Xmat\epmat ) | \Xmat, \hatmat \Phi )\\
		& = \Smat_{\epvec} \otimes \hatmatd{V} \hatmatd{D}^{-2} \hatmatd \Phi^2 (\hatmatd \Phi + \lambda \bm I_d )^{-2} \hatmatd{V}'.
	\end{align*}
	Combining these expressions, we can derive the following expression for the reducible prediction risk of the EgReg estimator:
	\begin{multline}
		R_r(\hat \bvec^E(d, \lambda)|\Xmat,\hatmat \Phi)
		= \tr\{  \Smat_{\epvec} \} \tr \{ \hatmatd{V} \hatmatd{D}^{-2} \hatmatd \Phi^2 (\hatmatd \Phi + \lambda \bm I_d )^{-2} \hatmatd{V}' \Smat_{\xvec} \} \\
		+ \lambda^2 \vect(\bvec^*)' \left( \bm I_q \otimes \hatmatd{V}(\hatmatd \Phi + \lambda \bm I_d )^{-1} \hatmatd{V}' \Smat_{\xvec} \hatmatd{V}(\hatmatd \Phi + \lambda \bm I_d )^{-1} \hatmatd{V}' \right) \vect(\bvec^*). \label{eqn:EgReg prediction risk}
	\end{multline}

	Moving on to NIECE, we find that
	the reducible bias is given by
	$
		\E[\hat \bvec^N | \Xmat, \hatmat \Phi] - \bm P_{\hatmatd{V}} \bvec^* = [\hatmat{V}_{(u)} \hatmat{V}_{(u)}' - \hatmatd{V} \hatmatd{V}'] \bvec^* . 
	$
	Note that if $u = d$, then NIECE does not reduce the predictors to a space any smaller than $\spn(\hatmatd{V})$ and its reducible bias is $0$.
	Moving on to the variance, we see that 
	$
		\Var( \vect(\hat{\bvec}^N ) | \Xmat, \hatmat \Phi)
		= \Var( \vect(  \hatmat{V}_{(u)} \hatmat{D}_{(u)}^{-2} \hatmat{V}_{(u)}' \Xmat'\epmat ) | \Xmat )
		= \Smat_{\epvec} \otimes \hatmat{V}_{(u)} \hatmat{D}_{(u)}^{-2} \hatmat{V}_{(u)}' \Xmat' \Xmat \hatmat{V}_{(u)} \hatmat{D}_{(u)}^{-2} \hatmat{V}_{(u)}'
		= \Smat_{\epvec} \otimes \hatmat{V}_{(u)} \hatmat{D}_{(u)}^{-2} \hatmat{V}_{(u)}'
	$
	where the final equality holds because $\Xmat' \Xmat = \hatmat{V} \hatmat{D}^2 \hatmat{V}'$ and $\hatmat{V}_{(u)}' \hatmat{V} \hatmat{D}^2 \hatmat{V}' \hatmat{V}_{(u)} = \hatmat{D}_{(u)}^2$. 
	All together, the reducible prediction risk for NIECE is given by:
	\begin{multline}
		R_r(\hat \bvec^N(u)|\Xmat,\hatmat \Phi)
		= \tr\{ \Smat_{\epvec} \} \tr\{ \hatmat{V}_{(u)} \hatmat{D}_{(u)}^{-2} \hatmat{V}_{(u)}' \Smat_{\xvec} \}\\
		+ \vect(\bvec^*)' \left( \bm I_q \otimes [\hatmat{V}_{(u)} \hatmat{V}_{(u)}' - \hatmatd{V} \hatmatd{V}'] \Smat_{\xvec}[\hatmat{V}_{(u)} \hatmat{V}_{(u)}' - \hatmatd{V} \hatmatd{V}'] \right) \vect(\bvec^*). \label{eqn:NIECE prediction risk}
	\end{multline}
	
	\subsubsection{Reducible prediction risk comparison}
	The following theorem directly compares the reducible prediction risks for EgReg and NIECE for fixed $n$ and $p$:
	\begin{thm} \label{thm:egreg vs niece prediction risk}
		$\forall_{n>0, p>0, u \leq r}, \exists_{d \geq u, \lambda > 0}$ such that  $R_r(\hat \bvec^E(d, \lambda)|\Xmat,\hatmat \Phi) < R_r(\hat \bvec^N(u)|\Xmat,\hatmat \Phi)$.
	\end{thm}
	
	Theorem \ref{thm:egreg vs niece prediction risk} establishes that for NIECE with any $u$, there always exists a pair $(d, \lambda)$ such that EgReg has lower reducible prediction risk.
	As we have noted, we are not concerned with whether the risk-minimizing $d$ for EgReg corresponds to the true envelope dimension, but are instead using shrinkage along the $d \geq u$ PCs to reduce the variance and the prediction risk.
	
	We note that, combined with \eqref{eqn:reducible irreducible risk decomp}, Theorem \ref{thm:egreg vs niece prediction risk} implies that EgReg admits a smaller upper bound for its prediction risk than NIECE.
	In addition, we highlight that both EgReg and NIECE can be computed when $p > n$ and that the bound in Theorem \ref{thm:egreg vs niece prediction risk} holds regardless of the alignment of $n$ and $p$.
	Moreover, Theorem \ref{thm:egreg vs niece prediction risk} allows the true model \eqref{eqn:linear model} to be either sparse or non-sparse.
	Thus EgReg can improve on the prediction performance of NIECE, a thresholding-based envelope estimator, under quite general settings for \eqref{eqn:linear model}.
	
	\subsection{Limiting prediction risk}
	To complement our finite-sample results, we derive the limiting prediction risks for EgReg and NIECE.
	This is, to our knowledge, the first time the limiting prediction risk of any predictor envelope estimator has been studied. 
	
	For the remainder of our theoretical study, we assume that $\Xmat = \Zmat \Smat_{\xvec}^{1/2}$ where the entries of $\Zmat$ are i.i.d for a distribution with zero mean and unit variance.
	The spectral decomposition of $\Smat_{\xvec}$ yields $\Smat_{\xvec} = \bm V \bm D^2 \bm V'$ where $\bm V = [\bm v_1 \cdots \bm v_p] \in \reals^{p \times p}$, the matrix of orthogonal eigenvectors of $\Smat_{\xvec}$, and $\bm D^2$ is the diagonal matrix of the eigenvalues of $\Smat_{\xvec}$. %

	In the envelope literature, some studies assume that certain envelope-related parameters are known to aid in understanding fundamental theoretical principles \citep{Cook2013,Cook_foundations_2015,Cook_simultaneous_2015}. 
	To facilitate our study of the limiting risks, we assume that $\Smat_{\xvec}$ and $\bm \Phi$ are known. 		
	As \citet{Zhang2023} have shown, the envelope dimension satisfies $u^* = \#\{\phi_j > 0\}$ %
	and $\matsub{V}{u^*}$ forms a basis for the predictor envelope $\envlp$.
	This implies that $\bvec^* \in \spn\{\matsub{V}{u^*}\}$.
	We further assume that $\matsub{\Phi}{u^*} = \matsub{D}{u^*}^2 = \bm I_{u^*}$ (note that this is weaker than the isotropic assumption $\Smat_{\xvec} = \bm I_p$, as it only constrains the eigenvectors of $\Smat_{\xvec}$ that comprise the envelope).

	\subsubsection{Prediction risks for EgReg and NIECE}
	We decompose the risk for an estimator $\hat \bvec$ of $\bvec^*$ into bias and variance components:
	\begin{align*}
		R(\hat \bvec|\Xmat) 
		& = \norm{  \vect ( \E[\hat \bvec | \Xmat] - \bvec^*) }^2_{\bm I_r \otimes \Smat_{\xvec}} 
		+ \tr\{\Var(\vect (\hat \bvec) | \Xmat)(\bm I_r \otimes \Smat_{\xvec}) \} \\
		& = B(\hat \bvec) + V(\hat \bvec)
	\end{align*}
	(note that $B(\hat \bvec)$ here differs from the reducible bias in \eqref{eqn:generic bv decomp}).
	
	Define $\Gmat = \matsub{V}{u^*} \matsub{D}{u^*}^{-1} \matsub{\Phi}{u^*}^{1/2} $. 
	The EgReg estimator (with $\Smat_{\xvec}$ and $\bm \Phi$ known) is given by
	$
		\hat \bvec^{E}(\lambda) = \Gmat( \Gmat'\Xmat'\Xmat\Gmat + n\lambda \bm I )^{-1} \Gmat' \Xmat'\Ymat.
	$
	Since $\matsub{\Phi}{u^*} = \matsub{D}{u^*}^2 = \bm I$, we see that $\Gmat'\Gmat = \bm I$.
	Combining this with the fact that $\bvec^* \in \spn\{ \Gmat \}$, we have $\bvec^* = \Gmat \Gmat'\bvec^*$.
	Using this fact, one can show
	$
		B(\hat \bvec^E(\lambda)) = \lambda^2 \tr\{ \Gmat'\bvec^* {\bvec^*}'\Gmat (\bm S_{\Xmat \Gmat} + \lambda \bm I)^{-2} \},
	$
	where $\bm S_{\Xmat \Gmat} = \frac{1}{n}  \Gmat'\Xmat'\Xmat\Gmat $.
	It is similarly straightforward to derive
	$
		V(\hat \bvec^E(\lambda)) = \frac{1}{n} \tr\{ \Smat_{\epvec} \} \tr\{ (\bm S_{\Xmat \Gmat} + \lambda \bm I)^{-2}\bm S_{\Xmat \Gmat}  \}.
	$
	
	When $u^* \leq n$, the NIECE estimator (with $\Smat$ and $\bm \Phi$ known) can be written as
	$
		\hat \bvec^{N} = \matsub{V}{u^*}(\matsub{V}{u^*}' \Xmat' \Xmat \matsub{V}{u^*})^{-1} \matsub{V}{u^*}\Xmat'\Ymat.
	$
	Because $\bvec^* \in \spn( \matsub{V}{u^*} )$ and the columns of $\matsub{V}{u^*}$ are orthogonal, we have $\bvec^* = \matsub{V}{u^*}\matsub{V}{u^*}'\bvec^*$. 
	Using this fact, it is straightforward to show that the bias of the NIECE estimator is $B(\hat \bvec^N) = 0$.
	One can further show that
	$
		V( \hat \bvec^N) = \frac{1}{n} \tr\{\Smat_{\epvec}\} \tr \{\Smat_{\Xmat \matsub{V}{u^*}} \bm S_{\Xmat \matsub{V}{u^*}}^{-1}  \} 
	$
	where $\Smat_{\Xmat \matsub{V}{u^*}} = \Var(\Xmat \matsub{V}{u^*}) = \matsub{V}{u^*}'\Smat \matsub{V}{u^*} = \matsub{D}{u^*}^2$
	and $\bm S_{\Xmat \matsub{v}{u^*}} = \frac{1}{n} \matsub{V}{u^*}' \Xmat' \Xmat \matsub{V}{u^*} $.
	Since $\matsub{D}{u^*}^2 = \bm I_u$, this further simplifies to
	$
	V( \hat \bvec^N) = \frac{1}{n} \tr\{\Smat_{\epvec}\} \tr \{\bm S_{\Xmat \matsub{V}{u^*}}^{-1}  \} 
	$
	When $u^* > n$, we take the NIECE estimator to be the limit of the EgReg estimator as $\lambda \to 0^+$---that is,
	$
	\hat \bvec^{N} = \lim_{\lambda \to 0^+} \hat \bvec^E(\lambda)
	$. %
	
	\subsubsection{Limiting risk analysis}
	Having derived finite-sample expressions for the prediction risks of EgReg and NIECE, we derive their limiting prediction risks.
	
	\begin{thm} \label{thm:limiting risk}
		Suppose that $u^*$ and $n$ diverge, with $u^*/n \to \gamma \in (0, \infty)$.
		Suppose that $\matsub{\Phi}{u^*} = \matsub{D}{u^*}^2 = \bm I_u$, that the predictors $\xvec_i$ are i.i.d with mean $0$ and $\E[|\xvec_i \matsub{V}{u^*}|^{8 + \delta}] < \infty$ for some $\delta>0$, and that $\tr\{ \Gmat'\bvec^*{\bvec^*}'\Gmat \} = c^2$ for all $n$.
		Then, almost surely, the limiting prediction risk for NIECE is
		\begin{equation}
			R(\hat \bvec^N |\Xmat) \to
			\begin{cases}
				\tr\{ \Smat_{\epvec} \}\frac{\gamma}{1 - \gamma}  & \text{ for } \gamma < 1\\
				c^2(1 - \frac{1}{\gamma})+ \tr\{ \Smat_{\epvec} \} \frac{1}{\gamma - 1} & \text{ for } \gamma > 1
			\end{cases} \label{eqn:limiting NIECE risk}
		\end{equation}
		and the limiting prediction risk for EgReg is
		\begin{equation}
			R(\hat \bvec^E(\lambda)|\Xmat) \to c^2 \lambda^2 m'(-\lambda) + \tr\{ \Smat_{\epvec} \} \gamma( m(-\lambda) - \lambda m'(-\lambda) )
		\end{equation}
		where 
		$
		m(z) = [ 1 - \gamma - z - \sqrt{ (1 - \gamma - z)^2 - 4 \gamma z } ]/(2\gamma z)
		$.
		Moreover, the limiting EgReg risk is minimized at $\lambda^* = \tr\{ \Smat_{\epvec} \} \gamma/c^2$, with
		$
		R(\hat \bvec^E(\lambda^*)|\Xmat) \to \tr\{ \Smat_{\epvec} \} \gamma m(-\lambda^*)
		$.
	\end{thm}
	
	\begin{figure}
		\centering
		\caption{Limiting prediction risk}	
		\label{fig:thm_dd}
		\includegraphics[width=0.9\textwidth]{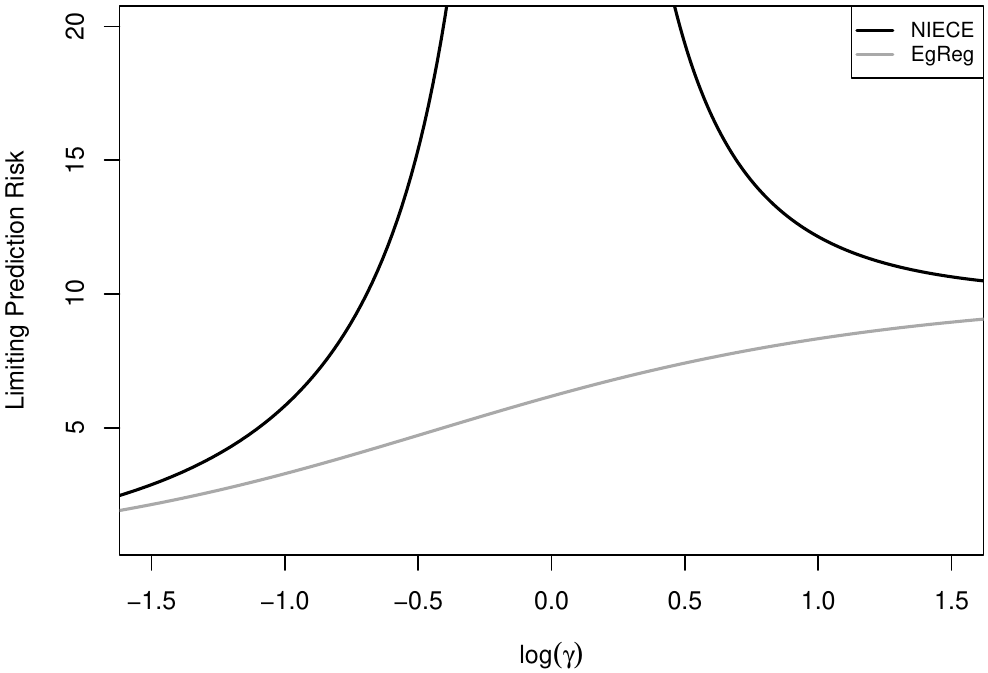}
	\end{figure}
	
	Figure \ref{fig:thm_dd} displays the limiting prediction risk of EgReg at the optimal $\lambda^\ast$ and NIECE, as described in Theorem \ref{thm:limiting risk}, for $\tr\{ \Smat_{\epvec} \}=10$ and $c^2=10$.
	We see that the limiting risk of NIECE rapidly diverges to $\infty$ as $\gamma \to 1$ from above and from below, exhibiting the double descent phenomenon seen in other predictive models \citep{Hastie2022}.
	Comparing the limiting risks in Theorem \ref{thm:limiting risk}, we see that for any $c^2$, $\tr\{\Smat_{\epvec} \}$, and $\gamma$ the limiting risk of EgReg at $\lambda^*$ is lower than that of NIECE.
	As such, EgReg delivers superior predictions to NIECE in the limit, particularly in the case where $u^* \approx n$.
	
	\section{Simulations}\label{sec:sims}
	In this section, we first compare the empirical prediction performance of EgReg with NIECE and a variety of other multivariate regression methods. 
	Next, we use simulations to verify the double descent phenomenon identified in our asymptotic theory. 
	We conclude the section by summarizing key takeaways.
		
	\subsection{Models and metric}
	We include the following methods in our comparison: 
	PLS (specifically the SIMPLS estimator \citep{de_jong_simpls_1993}); PCR; ridge regression; an envelope model fit to the PCs of $\Xmat$ (PCA Env) \citep{rimal2019comparison}; NIECE; Envelope-guided Regularization with both $d$ and $\lambda$ tuned (EgReg); and EgReg with $\lambda$ tuned and $d = r$, the rank of $\Xmat$ (EgReg(r)).
	We tune $\lambda$ for ridge and the EgReg models and $d/u$ for PLS, PCR, PCA Env, EgReg, and NIECE using 10-fold cross-validation (CV).
	
	We can divide these methods along a few axes: for instance, as thresholding methods (PLS, PCR, PCA Env, and NIECE) versus shrinkage methods (ridge and EgReg) or as envelope-based methods (PLS, PCA Env, NIECE, and EgReg) versus singular-value-based methods (PCR and ridge).
	We will pay particular attention to the comparison of EgReg and NIECE, as Theorems \ref{thm:egreg vs niece prediction risk} and \ref{thm:limiting risk} suggest EgReg should outperform NIECE.
	
	Throughout this study we assess the performance of the models based on the prediction risk. Let $\bvec^*$ denote the true coefficient vector. 
	The prediction risk can be expressed as
	$
	R(\hat \bvec|\Xmat)
	= \E[ \lVert \xvec_{new}'\hat \bvec - \xvec_{new}'\bvec^* \rVert_2^2 |\Xmat ]
	= \E[ \tr \{ (\hat \bvec - \bvec^*)' \Smat_{\xvec}(\hat \bvec - \bvec^*) \}|\Xmat ].
	$
	As such, we estimate the prediction risk over $R$ simulation replications with
	$$
	\hat R(\hat \bvec|\Xmat) = \frac{1}{R} \sum_{i = 1}^{R} \tr \{ (\hat \bvec^i - \bvec^*)' \Smat_{\xvec}(\hat \bvec^i - \bvec^*) \}
	$$
	where $\hat \bvec^i$ denotes the coefficient estimate for replication $i$.
	Note that we fix $\Xmat$ and $\bvec^*$ across replications, but vary the generated response $\Ymat$.
	
	\subsection{Predictor envelope model}\label{sim_predictor_envelope}
	For our first two sets of simulations, we generate data from a predictor envelope model,
	\begin{equation}
		\yvec' = \xvec' \Gmat \avec + \epvec' \text{,}
	\end{equation}
	where $\yvec, \epvec \in \reals^q$, $\xvec \in \reals^p$, $\xvec$ and $\epvec$ are independent, $\Gmat \in \reals^{p \times u^*}$ is a matrix with orthogonal columns, and $\avec \in \reals^{u^* \times q}$.
	We generate the data as follows:

	\textit{Input:} $n, p, r, \gamma, \mathcal{P}, \avec, \Smat_{\epvec}$
	\begin{enumerate}
		\item Generate a random orthogonal matrix $\bm V \in \reals^{p \times p}$ of eigenvectors of $\Smat_{\xvec}$.%
		\item Set $\sigma_i = 10 \exp(-\gamma(i - 1))$ for $i = 1, \ldots, p$. 
		Define $\bm D^2 = \diag\{ \sigma_i \}$.
		\item Set $\Smat_{\xvec} = \bm V \bm D^2 \bm V'$.
		\item Generate the predictors $\xvec_i \iid N(\bm 0, \Smat_{\xvec})$ for $i = 1, \ldots, n$.
		\item Set $\Gmat = \bm V_{ \mathcal{P} }$.
		\item Generate the response 
		$ \yvec_i' = \xvec_i' \Gmat \avec + \epvec_i' $,
		where $\epvec_{i} \iid N(\bm 0, \Smat_{\epvec})$ for $i = 1, \ldots, n$.
	\end{enumerate}

	\begin{remark}
		$\gamma$ controls how quickly the eigenvalues of $\Smat_{\xvec}$ decay. 
		$\mathcal{P}$ determines which eigenvectors of $\Smat_{\xvec}$ are basis vectors for the envelope subspace.
		As such, the true envelope dimension $u^* = |\mathcal{P}|$.
		Recall that the population-level envelope scores are given by $\phi_j = \bm v_j' \Smat_{\xvec \yvec} \Smat_{\yvec \xvec} '\bm v_j$.
		Since $\Smat_{\xvec \yvec} =  \Smat_{\xvec} \bvec^* = \Smat_{\xvec} \Gmat \avec = \Smat_{\xvec} \bm{V}_{\mathcal{P}} \avec = \bm{V}_{\mathcal{P}}\bm{D}_{\mathcal{P}}^2 \avec$, we find that
		$
			\phi_j 
			 = \bm v_j' \bm V_{ \mathcal{P} } \bm{D}_{\mathcal{P}}^2\avec \avec' \bm{D}_{\mathcal{P}}^2\bm V_{ \mathcal{P} }' \bm v_j
			 = \bm e_j' \bm{D}_{\mathcal{P}}^2 \avec \avec' \bm{D}_{\mathcal{P}}^2 \bm e_j
			 = \sigma_j^2 \avec_{\cdot j} \avec_{\cdot j}'
			 = \sigma_j^2\norm{\avec_{\cdot j} }_2^2
		$
		where $\bm e_j$ is the $j$th standard basis vector and $\avec_{\cdot j}$ denotes the $j$th row of $\avec$.
		As such, we can control the envelope scores through our choices of $\gamma$ and $\avec$.
	\end{remark}
	
	\begin{remark}
	We do not generate $\Gmat$ with any sort of sparsity structure.
	Rather, all of the predictors can feature into the reduced predictor $\xvec' \Gmat$ and be related to the response, making the true model for the simulations non-sparse.
	\end{remark}
	
	We vary $\mathcal{P}$ and $\avec$ to create different simulation settings.
	In particular, we examine how the location of the first index of $\mathcal{P}$ and the true envelope dimension affect the relative prediction performances of the models.
	Within each setting, we examine the cases where $p/n = 0.25, 0.5, 1, 2, 4$.
	Across all cases, we set $n=100$, $r = 1$, $\gamma = 1$, and $\Smat_{\epvec} = 10$.
	
	\subsubsection{$\mathcal{P}(1)$ study}
	In our first set of simulations, we vary the starting index of $\mathcal{P}$, examining the cases where $\mathcal{P}(1) = 1$ and $\mathcal{P}(1) = 7$. In both cases, we set $\bm V_{\mathcal{P}}$ to consist of ten consecutive eigenvectors, with $\mathcal{P} = \{\mathcal{P}(1), \mathcal{P}(1)+1, \ldots, \mathcal{P}(1)+9\}$, meaning that $u^*=10$, and set the reduced coefficient vector to be $\avec = (1, -1, 1, -1, \ldots)\in\mathbb R^{u^*}$.
	
	Figure \ref{fig:P1 sims} plots the prediction risks for each of the methods across $R = 100$ replications for a range of values of $p$.
	First, we note that the EgReg models outperform NIECE in both cases, providing empirical support for Theorem \ref{thm:egreg vs niece prediction risk}.
	We see that when $\mathcal{P}(1) = 1$ all of the models attain similar prediction risks.
	We would expect PCR and ridge to perform well in this setting---because the first $10$ eigenvectors of $\Smat_{\xvec}$ form a basis for the envelope, the PCs of $\Xmat$ with high variances also have high envelope scores, ensuring that ridge and PCR do not overly shrink/threshold important PCs or under shrink/fail to threshold unimportant ones.
	Given that this setting is particularly favorable to the non-envelope methods, it is notable that EgReg and the other envelope methods still deliver competitive prediction performances.
	
	\begin{figure}
		\centering
		\caption{$\mathcal{P}(1)$ study}
		\label{fig:P1 sims}	
		\includegraphics[width=0.9\textwidth]{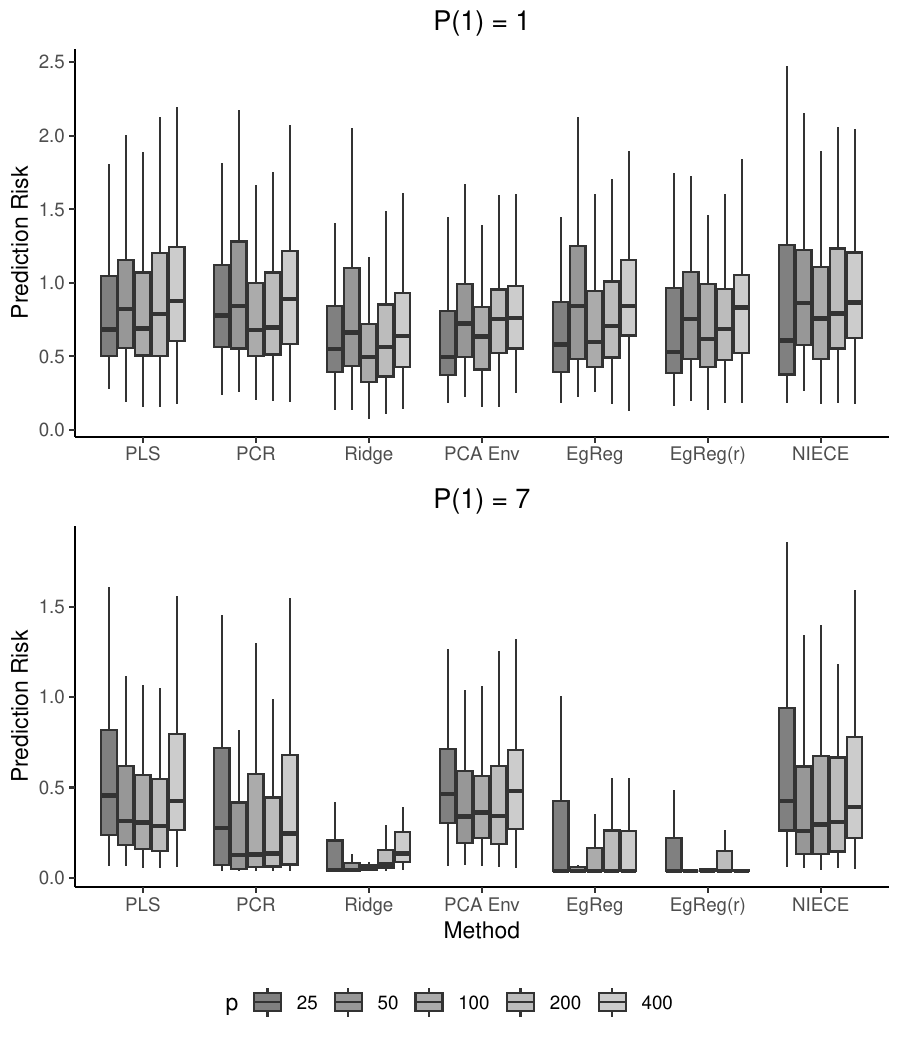}
	\end{figure}
	
	When $\mathcal{P}(1) = 7$ we see that the shrinkage methods (ridge and EgReg) dominate the thresholding methods.
	Notably, EgReg outperforms the envelope estimators, NIECE and PLS, by a wide margin.
	Moreover, we see that EgReg(r) tends to deliver a lower prediction risk than ridge when $p>25$ and that their performances are comparable when $p = 25$.
	Overall, this is a setting in which the EgReg models are clear favorites.
	This result aligns well with our intuition: EgReg is designed to do well when the PCs with the highest variances and those with the highest envelope scores do not align.
	Among the two EgReg models, the fixed-$d$ version (with $d=r$), EgReg(r), outperforms the tuned-$d$ version, EgReg, due to the latter's higher variance.
	This is likely a consequence of having to tune an additional parameter and could be ameliorated by employing a lower-variance tuning method such as leave-one-out CV.
	
	\subsubsection{Envelope dimension $u^*$ study}
	In our second set of simulations, we vary the true envelope dimension $u^* = |\mathcal{P}|$, examining the cases where $u^* = 5$ and $u^* = \min\{n,p\}/2$. 
	In each of these cases, we set $\mathcal{P}$ to consist of every other index, $\mathcal{P} = \{1,3,5,\ldots\}$, and set the entries of $\avec$ to be steadily increasing in magnitude from $0.1$ to $1$ and alternating in sign, with $\avec_j = (-1)^{(j \mod 2)}( 0.1 + j \frac{1 - 0.1}{u^*-1})$ for $j = 0, \ldots, u^*-1$.

	Figure \ref{fig:|P| sims} plots the prediction risks from this set of simulations.
	As in the previous simulations, the EgReg estimators deliver lower prediction risks than NIECE and PLS.
	Moreover, the gap between EgReg and these methods seems to grow as the envelope dimension $u^*$ increases.
	As in the $\mathcal{P}(1) = 7$ simulations, the shrinkage methods dominate the thresholding methods.
	This result again aligns with our intuition: Because the lower-variance PCs are more important in this model, discarding them significantly impacts the performance of the thresholding methods.
	We note that EgReg(r) also outperforms ridge in this setting.
	Thus this is yet another case where the EgReg models perform particularly well.
	
	\begin{figure}
		\centering
		\caption{Envelope dimension $u^*$ study}	
		\label{fig:|P| sims}
		\includegraphics[width=0.9\textwidth]{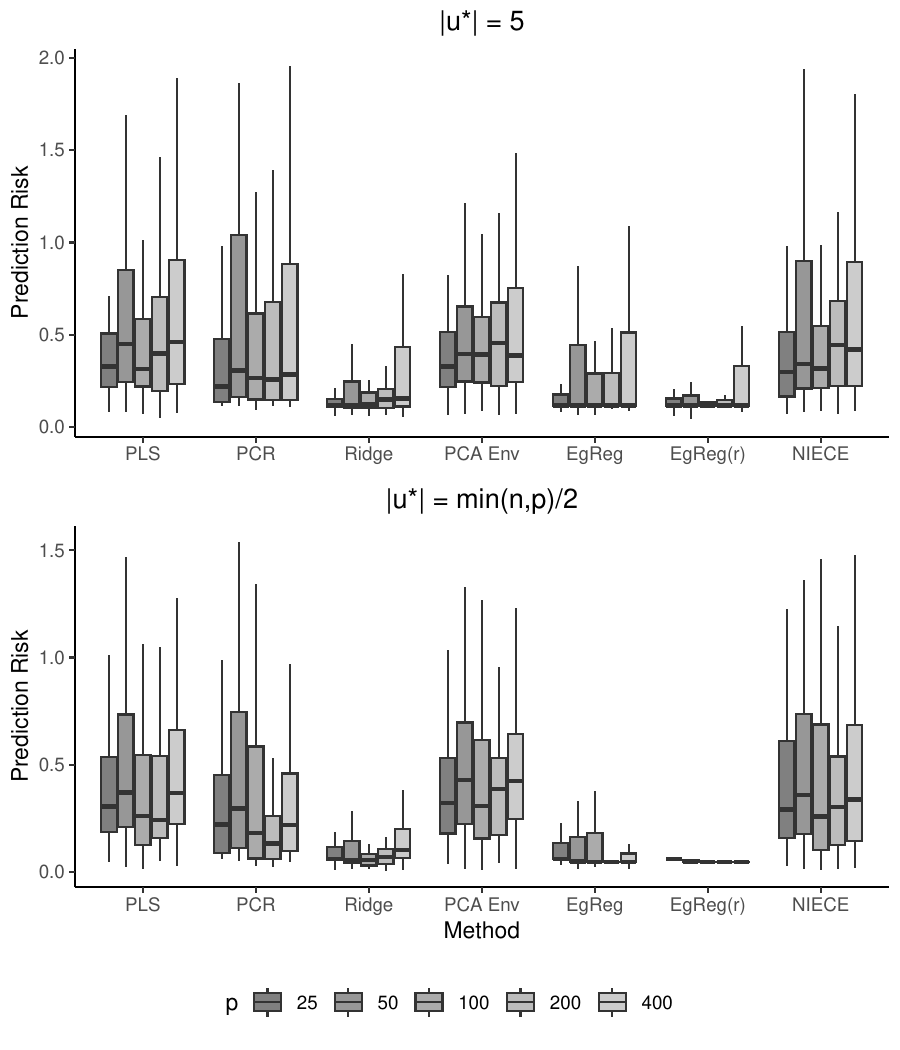}
	\end{figure}
	
	\subsection{No envelope structure}
	For our last set of simulations, we examine a few cases with no envelope structure.
	We generate the data from
	\begin{equation}
		\yvec' = \xvec'\bvec^* + \epvec'
	\end{equation}
	where  $\xvec_i \iid N(\bm 0, \Smat_{\xvec})$  and there is no special relationship between the eigenvectors of $\Smat_{\xvec}$ and $\bvec^*$: 
	We set $\Smat_{\xvec} = \rho \bm 1_{p\times p} + (1 - \rho) \bm I_p $ , the compound symmetric (CS) case, and $[\Smat_{\xvec}]_{ij} = \rho^{|i-j|}$ for $i,j \in [p]$, the auto-regressive (AR1) case.
	For both cases, we set $\rho = 0.5$ and $\bvec^* = (2,-2, 1,-1, 1/2, -1/2, 0, \ldots, 0)$.
	As before, we examine the cases where $p/n = 0.25, 0.5, 1, 2, 4$, with $n=100$.
	Note that only six coefficients are nonzero in the true model, making these sparse regressions.
	Our goal here is to examine how the EgReg models perform in baseline regression scenarios.
	
	Figure \ref{fig:base case simts} plots the prediction risks from the AR1 and CS cases.
	In the AR1 case, the shrinkage methods again hold an edge over the thresholding methods, with ridge and EgReg(r) performing particularly well.
	The pattern is similar in the CS case, though ridge appears to have a slight edge over EgReg and PLS performs relatively well when $p \leq 50$.
	As in the predictor envelope simulations, EgReg and EgReg(r) consistently attain lower risks than NIECE.
	Overall, the performance of the EgReg models, particularly EgReg(r), holds up in these baseline scenarios with no special envelope structure. 
	
	\begin{figure}
		\centering
		\caption{Baseline scenarios}	
		\label{fig:base case simts}
		\includegraphics[width=0.9\textwidth]{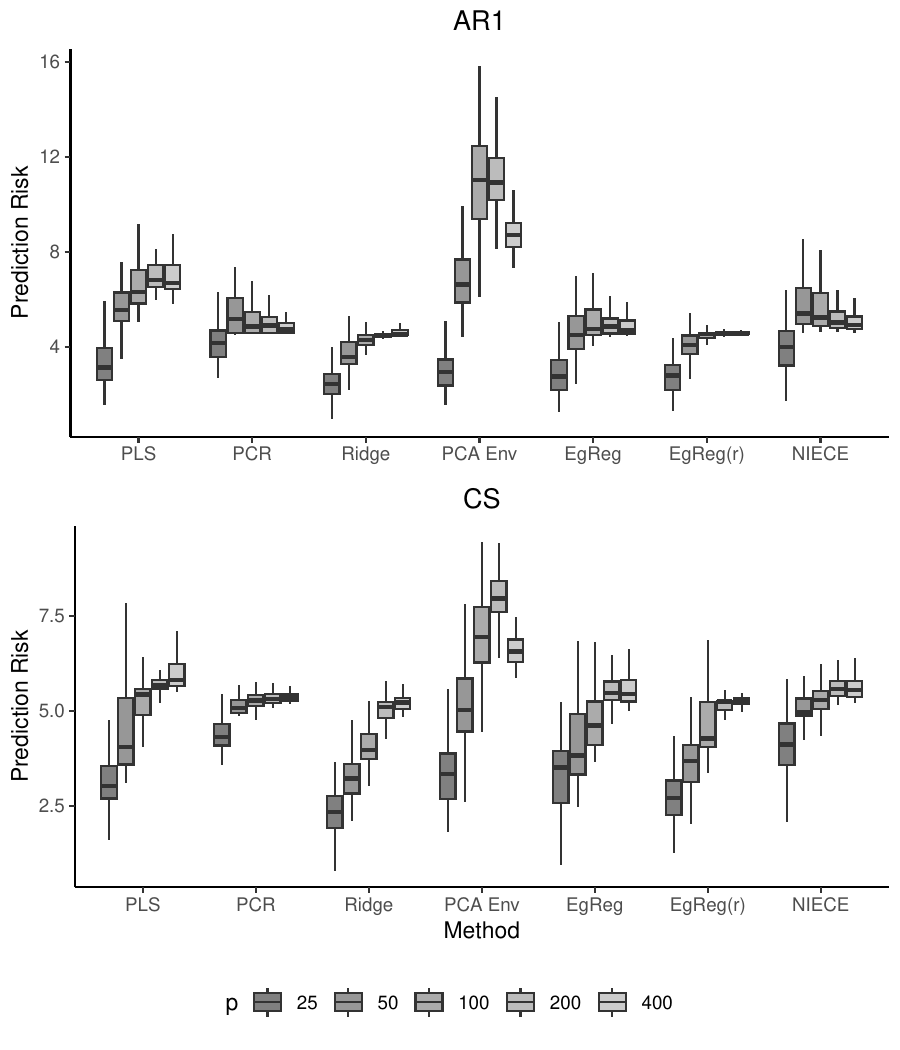}
	\end{figure}
	
	\subsection{Double descent confirmation}
	This subsection uses simulations to explore the double descent phenomenon of the NIECE estimator, as described in Theorem \ref{thm:limiting risk}, and compare the empirical prediction risks of NIECE and EgReg. 
	While the theorem assumes that $\Smat_{\xvec}$ and $\bm \Phi$ are known, in this simulation we treat these parameters as unknown and estimate them, reflecting the real-world scenario where such information is typically unavailable. 
	We examine cases where $u^\ast/n$ varies from $0.2$ to $5$.%
	
	Data are generated from a predictor envelope model using the procedure and setup outlined in Section \ref{sim_predictor_envelope}, with the modification that we set $p = \left\lfloor 1.5\cdot u^\ast \right\rceil$ and $\sigma_i=1$ for $i=1,\ldots,p$. 
	The starting index is set as $\mathcal P(1)=7$ and the index set is defined as $\mathcal P = \{ \mathcal P (1), \mathcal P (1) +1, \ldots, \mathcal P (1) + u^\ast - 1\}$. 
	The reduced coefficient vector is specified as $\avec = \sqrt{10} \cdot \bm \eta / \| \bm \eta \|_2 $ where $\bm \eta = (1, -1, 1, -1, \ldots) \in \mathbb R^{u^\ast}$, ensuring that $c^2=10$. 
	NIECE is computed with $u=\min(u^\ast, n-1)$ (we set $u \leq n-1$ to ensure numerical stability). 
	
	Figure \ref{fig:sim_dd} illustrates the prediction risks for NIECE, EgReg, and EgReg(r), computed across $R = 100$ replications for varying values of $u^\ast$. 
	The empirical prediction risks mirror the theoretical limiting risks from Figure \ref{fig:thm_dd}.
	First, we observe the double descent phenomenon exhibited by the NIECE model. 
	Second, EgReg achieves smaller prediction risks across the range of $u^\ast$ values. 
	These observations align with the implications of Theorem \ref{thm:limiting risk}. Finally, EgReg and EgReg(r) exhibit similar performances, with EgReg(r) showing slightly better prediction risk in general.
	
	\begin{figure}
		\centering
		\caption{Double descent confirmation}	
		\label{fig:sim_dd}
		\includegraphics[width=0.9\textwidth]{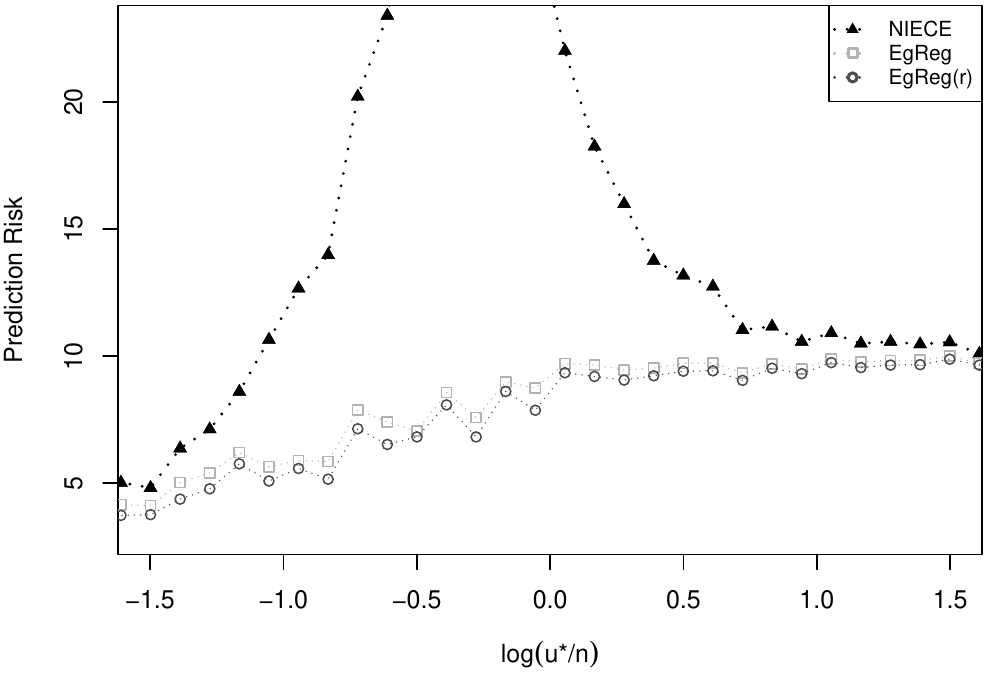}
	\end{figure}
		
	\subsection{Simulation takeaways}
	To summarize the main findings from our simulations, we found that there are settings in which the EgReg models dominate, particularly when there is an underlying predictor envelope structure to the data and some of the PCs which comprise the envelope have relatively low variances.
	Moreover, we saw that even when there is no envelope structure, the EgReg models deliver lower prediction risks than thresholding methods and comparable risks to ridge. 
	Lastly, we observed that NIECE exhibits a double descent risk curve plotted over $u^*/n$ and that the EgReg methods consistently provide lower prediction risks than NIECE over a range of $u^*$ values. 
	The advantage of using EgReg over NIECE is particularly pronounced when $u^*$ and $n$ are close.	
	Together these findings suggest that the EgReg estimators are always worth trying in multivariate prediction problems: They are competitive in every setting we examined and dominate in some cases.
	
	\section{Real Data Analysis}\label{sec:real data}
	In chemometrics, near-infrared (NIR) spectroscopy is widely used to predict chemical properties because it is fast, accurate, and requires minimal sample preparation. Compared to conventional laboratory methods, which are often time-consuming, labor-intensive, and require large amounts of reagents, NIR spectroscopy is much more efficient. 
	
	NIR spectroscopy data of manure samples from poultry and cattle were collected by \cite{goge2021dataset}. The data are available in the Data INRAE Repository at \href{https://doi.org/10.15454/JIGO8R}{\nolinkurl{https://doi.org/10.15454/JIGO8R}}. These data contain 332 manure samples that were analyzed by NIR spectroscopy using a NIRFlex device. NIR spectra were recorded every 2 nm from 1100 to 2498 nm on fresh homogenized samples. Additionally, the cattle manure samples were analyzed for chemical properties, including the amount of dry matter. We use these data to predict dry matter from the absorbance spectra. 
	
	In our analysis, we consider a linear model, where $y_i \in \mathbb{R}$ is the chemical measurement of dry matter and $\xvec_i \in \mathbb{R}^{p}$ is the vector of NIR spectroscopy measurements. To illustrate the behavior of EgReg's prediction performance with different numbers of predictors, we consider several values of $p$ by slicing the spectral data at different intervals: 

	\begin{itemize}
		\item \textbf{Case 1:} $p=70$; NIR spectra recorded every 20nm from 1100 to 2498 nm.
		\item \textbf{Case 2:} $p=140$; NIR spectra recorded every 10nm from 1100 to 2498 nm.
		\item \textbf{Case 3:} $p=350$; NIR spectra recorded every 4nm from 1100 to 2498 nm.
		\item \textbf{Case 4:} $p=700$; NIR spectra recorded every 2nm from 1100 to 2498 nm.
	\end{itemize}
	The first two cases correspond to low-dimensional settings and the last two cases represent high-dimensional settings. Before the analysis, we standardize each spectroscopy and dry matter measurement to have a sample mean of 0 and a standard deviation of 1. 
	
	To compare the prediction performance, we randomly take 80\% of the samples (i.e., training samples of 266) from the data to tune and fit models. 
	We perform 10-fold CV to select tuning parameters for each model. 
	We select $d$ from $\{1,2,\ldots,\min(p,230)\}$ for PLS and PCR and $u$/$d$ from $\{2,4,\ldots,\min(p,230)\}$ for PCA Env, NIECE, and EgReg for computational efficiency. 
	For the ridge and EgReg models, we select $\lambda$ from 50 equally-spaced values on the log-scale using 10-fold CV.
	
	The remaining 20\% of the samples (i.e., test samples of 66) are used as the test set. On the test set, we calculate the relative prediction error (RPE) compared to PLS, which is often regarded as the most typical method in chemometrics, using the following formula:
	$
	\text{RPE} = \sum_{i=1}^{66} (y_i^{\text{test}} - \hat y_i^{\text{test}} )^2 / \sum_{i=1}^{66}  (y_i^{\text{test}} - \hat y^{\text{test}}_{i,\text{PLS}})^2 . 
	$
	We repeat this process 100 times and report the average relative prediction error (ARPE) and its standard error (SE) in Table \ref{tab:real1}. 
	The results show that the EgReg and EgReg(r) methods yield the smallest prediction errors. 
	It is worth highlighting that EgReg and EgReg(r) outperform PLS in every case. 
	Moreover, ridge, the main competitor for EgReg in our simulation studies, does far worse than EgReg in making predictions on these real data.
	
	\begin{table}[h]
		\caption{Near-infrared spectroscopy data: Average relative prediction error (ARPE) and its standard error (SE) (lowest ARPE for each case in \textbf{bold}).}\label{tab:real1}
		\begin{center}
			\small
			\begin{tabular}{c|cccccc} \toprule
				& PCR & PCA Env & NIECE & Ridge & EgReg &  EgReg(r) \\  \hline
				\multicolumn{7}{l}{\textbf{Case 1:} $p = 70$ (low-dimensional).} \\ \hline
				ARPE & 0.978 & 2.870 & 0.994 & 1.244 & \textbf{0.974} & \textbf{0.974} \\
				SE   & 0.004 & 0.042 & 0.004 & 0.014 & 0.004 & 0.003 \\ \hline
				\multicolumn{7}{l}{\textbf{Case 2:} $p = 140$ (low-dimensional).} \\ \hline
				ARPE & 0.988 & 2.873 & 0.991 & 1.279 & 0.975 & \textbf{0.971} \\
				SE   & 0.006 & 0.042 & 0.005 & 0.016 & 0.004 & 0.004 \\ \hline
				\multicolumn{7}{l}{\textbf{Case 3:} $p = 350$ (high-dimensional).} \\ \hline
				ARPE & 0.997 & 2.889 & 0.997 & 1.609 & 0.967 & \textbf{0.960} \\
				SE   & 0.005 & 0.042 & 0.005 & 0.032 & 0.004 & 0.005 \\ \hline
				\multicolumn{7}{l}{\textbf{Case 4:} $p = 700$ (high-dimensional).} \\ \hline
				ARPE & 0.999 & 2.882 & 1.002 & 1.991 & 0.970 & \textbf{0.962} \\
				SE   & 0.005 & 0.042 & 0.005 & 0.026 & 0.005 & 0.004 \\ \bottomrule
			\end{tabular}
		\end{center}
	\end{table}
	
	\section{Discussion}\label{sec:discussion}
	In this paper, we proposed EgReg, an envelope-based shrinkage estimator for improved prediction in multivariate regression.
	We have shown, theoretically and empirically, that EgReg achieves superior prediction accuracy compared to NIECE, a closely related non-shrinkage predictor envelope estimator, in both finite dimensions and as $u^*/n \to \gamma \in (0, \infty)$.
	In simulations and a real data analysis, we have found that EgReg can provide more accurate predictions than popular multivariate prediction methods---PCR, ridge, and PLS---in both low and high dimensions.

	EgReg is, to our knowledge, the first estimator to use envelopes to guide shrinkage.
	As such, there are numerous ways to extend the EgReg framework in future research.
	EgReg applies a ridge-like penalty to the reduced predictor coefficient $\evec$ to improve the prediction accuracy.
	We could generalize the EgReg objective \eqref{eqn:reduced ridge problem} to
	$
		\hat \evec = \argmin_{\eta \in \reals^d } \norm{\Ymat - \Xmat \hatmat \Gamma \evec}_F^2 + p_{\lambda}(\evec),
	$
	where $p_{\lambda}(\cdot)$ is another penalty, such as the lasso \citep{Tibshirani1996}, elastic net \citep{Zou2005}, or group lasso \citep{yuan_model_2006}.
	Alternatively, we could use a different ``greedy'' estimate of the envelope basis, replacing $\hatmat \Gamma = \hatmatd{V} \hatmatd{D}^{-1}\hatmatd\Phi^{1/2}$ with, for example, the PLS-estimated envelope basis \citep{Cook2013} or the penalized matrix decomposition estimate of $\bm V_d$ \citep{Witten2009}. 
	Theoretically, we could derive the rate of the EgReg predictions as \citet{cook_partial_2019} did for PLS, particularly in the context of abundant regression where both $p$ and the covariance between the predictors and the response increase at the same rate. 
	Investigating EgReg under this regime would be an interesting direction for future research as EgReg outperformed PLS in our simulations and real data analysis.
	Lastly, we could generalize the EgReg framework for envelope-guided shrinkage in other predictive models, such as Huber regression and quantile regression.
	
	\if0\withappendix{
		\bigskip
		\begin{center}
			{\large\bf SUPPLEMENTARY MATERIAL}
		\end{center}
		
		\begin{description}
			\item[egreg\_supp:] Contains proofs of the theoretical results in the paper. (.pdf)
			
			\item[egrep\_repro:] Contains R code to reproduce the empirical results in the paper. (.zip)
		\end{description}
	}\fi
	
	\bibliographystyle{apalike}
	\bibliography{./envlp}
	
	\if1\withappendix{
		
		\appendix
		
		\section{Theoretical Proofs} \label{sec:proofs}
		\begin{proof}[Proof of Theorem \ref{thm:egreg vs niece prediction risk}]
			We start by showing that when EgReg and NIECE use the same number of PCs $u$, there always exists $\lambda > 0$ such that 
			\begin{equation}
				R_r(\hat \bvec^E(u, \lambda)|\Xmat,\hatmat \Phi) < R_r(\hat \bvec^N(u)|\Xmat,\hatmat \Phi) . \label{eqn:fixed u risk comparison}
			\end{equation}
			We see that when $\lambda = 0$, 
			$\hat \bvec^E(u, 0) = \hatmatd{V} \hatmatd{D}^{-1} \hatmatd \Phi \hatmatd \Phi^{-1} \hatmatd{U}'\Ymat = \hat \bvec^N(u)$.
			As such,\\ $R_r(\hat \bvec^E(u, 0)|\Xmat,\hatmat \Phi) = R_r(\hat \bvec^N(u)|\Xmat,\hatmat \Phi)$.
			Therefore we need only show that there exists $\lambda' > 0$ such that $\frac{\partial}{\partial \lambda} R_r(\hat \bvec^E(u,\lambda)|\Xmat,\hatmat \Phi) < 0$ for all $\lambda \in (0, \lambda')$ to establish \eqref{eqn:fixed u risk comparison}. 
			
			Taking the derivative of \eqref{eqn:EgReg prediction risk} with respect to $\lambda$, one can show that
			\begin{multline*}
				\frac{\partial}{\partial \lambda} R_r(\hat \bvec^E(u,\lambda)|\Xmat,\hatmat \Phi)
				= -2 \tr\{  \Smat_{\epvec} \} \tr \{\hatmatd{V}' \Smat_{\xvec} \hatmatd{V} \hatmatd{D}^{-2} \hatmatd \Phi^2 (\hatmatd \Phi + \lambda \bm I_d )^{-3} \} \\
				+ 2\lambda \tr \{ \bvec^* {\bvec^*}' \hatmatd{V} \hatmatd \Phi (\hatmatd \Phi + \lambda \bm I_d )^{-2} \hatmatd{V}' \Smat_{\xvec} \hatmatd{V} (\hatmatd \Phi + \lambda \bm I_d )^{-1} \hatmatd{V}' \} .
			\end{multline*}
			For a matrix $\bm M$, we define $\sigma_j(\bm M)$ to be its $j$th largest eigenvalue. 
			By von Neumann's trace inequality and the cyclic property of eigenvalues, we find
			\begin{align*}
				\tr \{ \bvec^* {\bvec^*}' \hatmatd{V} & \hatmatd \Phi (\hatmatd \Phi + \lambda \bm I_d )^{-2} \hatmatd{V}' \Smat_{\xvec} \hatmatd{V} (\hatmatd \Phi + \lambda \bm I_d )^{-1} \hatmatd{V}' \}\\
				& \leq \sum_{j=1}^p \sigma_j(\bvec^* {\bvec^*}') \sigma_j(\hatmatd{V} \hatmatd \Phi (\hatmatd \Phi + \lambda \bm I_d )^{-2} \hatmatd{V}' \Smat_{\xvec} \hatmatd{V} (\hatmatd \Phi + \lambda \bm I_d )^{-1} \hatmatd{V}')\\
				& \leq \sigma_1(\bvec^* {\bvec^*}') \sum_{j=1}^p  \sigma_j(\hatmatd \Phi(\hatmatd \Phi + \lambda \bm I_d )^{-3} \hatmatd{V}' \Smat_{\xvec} \hatmatd{V}) \\
				& = \sigma_1(\bvec^* {\bvec^*}') \sum_{j=1}^p  \sigma_j( \hatmatd \Phi^{-1} \hatmatd{D}^2  \hatmatd{D}^{-2}  \hatmatd \Phi^2(\hatmatd \Phi + \lambda \bm I_d )^{-3} \hatmatd{V}' \Smat_{\xvec} \hatmatd{V})\\
				& \leq \sigma_1(\bvec^* {\bvec^*}') \sum_{j=1}^p  \sigma_j( \hatmatd \Phi^{-1} \hatmatd{D}^2 )\sigma_j( \hatmatd{D}^{-2} \hatmatd \Phi^2(\hatmatd \Phi + \lambda \bm I_d )^{-3} \hatmatd{V}' \Smat_{\xvec} \hatmatd{V})\\
				& \leq \sigma_1(\bvec^* {\bvec^*}') \sigma_1( \hatmatd \Phi^{-1} \hatmatd{D}^2 ) \sum_{j=1}^p \sigma_j(\hatmatd{V}' \Smat_{\xvec} \hatmatd{V} \hatmatd{D}^{-2} \hatmatd \Phi^2(\hatmatd \Phi + \lambda \bm I_d )^{-3} )
			\end{align*}
			Applying this to the previous expression for the risk yields
			\begin{multline*}
				\frac{\partial}{\partial \lambda} R_r(\hat \bvec^E(u,\lambda)|\Xmat,\hatmat \Phi)
				\leq 2( \lambda \sigma_1(\bvec^* {\bvec^*}') \sigma_1( \hatmatd \Phi^{-1} \hatmatd{D}^2 ) -  \tr\{  \Smat_{\epvec} \}) \\
				\times \sum_{j=1}^p \sigma_j(\hatmatd{V}' \Smat_{\xvec} \hatmatd{V} \hatmatd{D}^{-2} \hatmatd \Phi^2(\hatmatd \Phi + \lambda \bm I_d )^{-3} )
			\end{multline*}
			Since $\hatmatd{V}' \Smat_{\xvec} \hatmatd{V}  \hatmatd{D}^{-2}  \hatmatd \Phi^2(\hatmatd \Phi + \lambda \bm I_d )^{-3}$ is positive semi-definite the sum is non-negative.
			As such,	if $\lambda < \tr\{  \Smat_{\epvec}\}/(\sigma_1(\bvec^* {\bvec^*}')\sigma_1( \hatmatd \Phi^{-1} \hatmatd{D}^2)) $, then $\frac{\partial}{\partial \lambda} R_r(\hat \bvec^E(u,\lambda)|\Xmat,\hatmat \Phi) < 0$ and \eqref{eqn:fixed u risk comparison} holds.
			
			To finish the proof, we note that for a fixed $\lambda$ there always exists $d\geq u$, namely $d = u$, such that $ R_r(\hat \bvec^E(d, \lambda)|\Xmat, \hatmat \Phi) \leq R_r(\hat \bvec^E(u, \lambda)|\Xmat,\hatmat \Phi)$.
			Combining this with \eqref{eqn:fixed u risk comparison} yields the desired result.
		\end{proof}
		
		\begin{proof}[Proof of Theorem \ref{thm:limiting risk}]
			Define $\Wmat =  \Xmat \matsub{V}{u^*}$. Under our assumptions $\Wmat = \Zmat \Smat_{\xvec}^{1/2} \matsub{V}{u^*} = \Zmat \Smat_{\bm w}^{1/2}$ and $\Smat_{\bm w} = \matsub{V}{u^*}'\Smat_{\xvec}\matsub{V}{u^*} = \matsub{V}{u^*}'\matsub{V}{u^*}\matsub{D}{u^*}^2\matsub{V}{u^*}'\matsub{V}{u^*} =\matsub{D}{u^*}^2 = \bm I $.
			In addition, $\Xmat \Gmat = \Xmat \matsub{V}{u^*}$.
			If $r = 1$, then the conditions of Theorem 1 and Corollary 5 of \citet{Hastie2022} are satisfied and the proof is complete.
			
			We now examine the case where $r > 1$, starting with EgReg.
			We have established that the bias is $B(\hat \bvec^E(\lambda)) = \lambda^2 \tr\{ \Gmat'\bvec^* {\bvec^*}'\Gmat (\bm S_{\Xmat \Gmat} + \lambda \bm I)^{-2} \}$.
			We define $\bm \Theta = \Gmat'\bvec^* {\bvec^*}'\Gmat$ and $\bm S = \bm S_{\Xmat \Gmat}$ to simplify notation. 
			Define $g_{n,\lambda}(\eta) = \lambda \tr\{ \bm \Theta (\bm S + \lambda (1 + \eta) \bm I)^{-1} \} $ for $\eta \in \mathcal D =  \{|\eta|<1\} \in \mathbb{C}$.
			We see that $- \frac{\partial}{\partial \eta} g_{n,\lambda}(0) = \lambda^2 \tr\{ \bm \Theta (\bm S + \lambda \bm I)^{-2} \} = B(\hat \bvec^E(\lambda))$.
			We will study the limiting behavior of $g_{n,\lambda}(\eta)$ and use that as a stepping stone to evaluate the limiting behavior of its derivative.
			
			Theorem 1 of \citet{Rubio2011} provides that
			\begin{equation}
				\tr\{\bm \Theta (\bm S + \lambda (1 + \eta) \bm I)^{-1}\} - \tr\{\bm \Theta (a_n \bm I + \lambda (1 + \eta) \bm I)^{-1}\} \asto 0 \label{eqn:det equiv 1}
			\end{equation}
			as $n,u^* \to\infty$, where $a_n$ is the unique solution to 
			$
			\frac{1}{\gamma}\left(\frac{1}{a_n} - 1\right) = \frac{1}{u^*}\tr\{ (a_n \bm I + \lambda (1 + \eta) \bm I)^{-1}\}.
			$
			We aim to evaluate the limit of $ \tr\{ \bm \Theta (a_n \bm I + \lambda (1 + \eta) \bm I)^{-1}\} = \frac{c^2}{a_n + \lambda (1 + \eta)}$.
			Applying Theorem 1 of \citet{Rubio2011}, we see that
			$
			\tr\{(\bm S + \lambda (1 + \eta) \bm I)^{-1}\} - \tr\{ (a_n \bm I + \lambda (1 + \eta) \bm I)^{-1}\} \asto 0
			$.
			At the same time, the Marchenko-Pasteur Theorem provides that
			$
			\frac{1}{u^*}\tr\{(\bm S + \lambda (1 + \eta) \bm I)^{-1}\} \asto m(-\lambda (1 + \eta))
			$
			where $m(\cdot)$ is the Stieltjes transform of the limiting spectral distribution of $\Smat_{\bm w}$.
			Combining these results yields $\frac{1}{a_n + \lambda(1 + \eta)} = \frac{1}{u^*}\tr\{ (a_n \bm I + \lambda (1 + \eta) \bm I)^{-1}\} \asto m(-\lambda(1 + \eta))$.
			By \eqref{eqn:det equiv 1}, this implies $g_n(\lambda)  \asto c^2 \lambda m(-\lambda(1 + \eta))$.
			
			We find $|g_{n,\lambda}(\eta) | \leq |c^2 \lambda[ \lambda(1 + \eta) ]^{-1} | \leq c^2 \lambda (\lambda/2)^{-1} \leq 2c^2$ for all $n, \eta \in \mathcal{D}$.
			In addition, we see that $g_{n,\lambda}(\eta)$ is analytic.
			As such, we can apply Lemma 2.14 of \citet{Bai2010} and exchange the derivative with respect to $\lambda$ and the limit with respect to $n,u^*$, yielding
			$
			\frac{\partial}{\partial \eta} g_{n,\lambda}(\eta) \asto -c^2 \lambda^2 m'(-\lambda(1 + \eta))
			$.
			Therefore $B(\hat \bvec^E(\lambda)) = - \frac{\partial}{\partial \eta} g_{n,\lambda}(0) \asto c^2 \lambda^2 m'(-\lambda) $.
			
			Moving on to the variance term, we find
			$
			\tr\{ (\bm S + \lambda \bm I)^{-2}\bm S  \}
			= \tr\{ (\bm S + \lambda \bm I)^{-2}(\bm S + \lambda \bm I - \lambda \bm I ) \} 
			= \tr\{ (\bm S + \lambda \bm I)^{-1} \} - \lambda \tr\{ (\bm S + \lambda \bm I)^{-2} \}
			$. 
			We have already shown that $\frac{1}{u^*} \tr\{ (\bm S + \lambda \bm I)^{-1} \} \asto m(-\lambda) $.
			By defining $h_{n,\lambda}(\eta) = \tr\{ (\bm S + \lambda (1 + \eta) \bm I)^{-1}\}$ and employing the same strategy as we did with $g_{n,\lambda}(\eta)$, one can show that $ \lambda \frac{1}{u^*} \tr\{ (\bm S + \lambda \bm I)^{-2} \} \asto \lambda m'(-\lambda).$
			Combining these results yields
			\begin{align*}
				V(\hat \bvec^E(\lambda)) 
				& = \frac{u^*}{n} \tr\{ \Smat_{\epvec} \} \left( \frac{1}{u^*}\tr\{ (\bm S + \lambda \bm I)^{-1} \} - \lambda \frac{1}{u^*} \tr\{ (\bm S + \lambda \bm I)^{-2} \} \right)\\
				& \asto \gamma \tr\{ \Smat_{\epvec} \} ( m(-\lambda) - \lambda m'(-\lambda) ),
			\end{align*}
			completing this portion of the proof.

			Next, we evaluate the limiting risk of NIECE. Our strategy will be to leverage the fact that $\hat \bvec^N = \lim_{0\to \lambda^+} \hat \bvec^E(\lambda)$ and exchange the limits with respect to $u^*,n\to \infty$ and $\lambda \to 0^+$ to derive the limiting prediction risk.
			We have three sequences of functions we need to take the limits of: $f_n(\lambda) = \lambda^2 \tr\{ \bm \Theta (\bm S + \lambda \bm I)^{-2} \}$, $g_n(\lambda) = \frac{1}{u^*}\tr\{ (\bm S + \lambda \bm I)^{-1} \} $, and $h_n(\lambda) = \lambda \frac{1}{u^*} \tr\{ (\bm S + \lambda \bm I)^{-2} \}$.
			
			We have already shown that $f_n(\lambda) \asto c^2 \lambda^2 m'(-\lambda)$ for all $\lambda > 0$.
			Let $\sigma_{\min}^+(\bm S)$ and $\sigma_{\max}^+(\bm S)$ denote the smallest and largest positive eigenvalues, respectively, of $\bm S$.
			We see that 
			$
			|f_n(\lambda)| \leq c^2 \frac{\lambda^2}{ (\sigma_{\min}^+ (\bm S) + \lambda)^2 } \leq c^2
			$ for all $n, \lambda$.
			One can show that $\frac{\partial }{\partial \lambda }f_n(\lambda) = 2\lambda \tr\{ \bm \Theta (\bm S + \lambda \bm I)^{-3} \bm S\}$.
			Theorem 1 of \citet{Bai1993} implies that $\sigma_{\min}^+(\bm S) > (1 - \sqrt{\gamma})^2/2$ and $\sigma_{\max}^+(\bm S) < 2(1 + \sqrt{\gamma})^2$ almost surely for sufficiently large $n$.
			Applying this result, we find $|\frac{\partial }{\partial \lambda }f_n(\lambda)| \leq 2 c^2 \lambda \frac{\sigma_{\max}^*(\bm S)}{( \sigma_{\min}^+(\bm S) + \lambda )^3} \leq 2 c^2  \frac{\sigma_{\max}^*(\bm S)}{( \sigma_{\min}^+(\bm S) + \lambda )^2} \leq 16 c^2 \frac{(1 + \sqrt{\gamma})^2 }{(1 - \sqrt{\gamma})^4} $ almost surely for all $\lambda>0$ and sufficiently large $n$.
			Because the sequence of derivatives $\{ \frac{\partial }{\partial \lambda }f_n(\lambda) \}$ is uniformly bounded, $\{ f_n(\lambda) \}$ is equicontinuous. 
			Combining this with the fact that $\{ f_n(\lambda) \}$ is uniformly bounded and converges almost surely, we can apply the Arzela-Ascoli Theorem to conclude that  $\{ f_n(\lambda) \}$ converges uniformly.
			As an immediate consequence, the Moore-Osgood Theorem enables us to exchange the limits with respect to $n,u^*\to \infty$ and $\lambda \to 0^+$ and conclude that
			$
			\lim_{n,u^* \to \infty} \lim_{\lambda \to 0^+} f_n(\lambda) = \lim_{\lambda \to 0^+} \lim_{n,u^* \to \infty}  f_n(\lambda) =  \lim_{\lambda \to 0^+} c^2 \lambda^2 m'(-\lambda)
			$.
			
			By the same argument, the limits with respect to $n,u^*\to \infty$ and $\lambda \to 0^+$ can  be exchanged for $\{ g_n(\lambda) \}$ and $\{ h_n(\lambda) \}$ as well.
			All together, we have that
			\begin{align*}
				\lim_{n,u^* \to \infty} R(\hat \bvec^N) 
				& = \lim_{n,u^* \to \infty} \lim_{\lambda \to 0^+} R(\hat \bvec^E(\lambda))
				= \lim_{\lambda \to 0^+} \lim_{n,u^* \to \infty} R(\hat \bvec^E(\lambda))\\
				& = \lim_{\lambda \to 0^+} c^2 \lambda^2 m'(-\lambda) + \tr\{ \Smat_{\epvec} \} \gamma( m(-\lambda) - \lambda m'(-\lambda) ).
			\end{align*}
			One can show that $m(z) = (1 - \gamma)^{-1} + O(z)$ for $\gamma < 1$ and $m(z) = -\frac{(\gamma - 1)}{\gamma z} + \frac{1}{(\gamma - 1) \gamma} + O(z)$ for $\gamma > 1$, where ``$f(z) = O(g(z))$'' means there exists $C > 0$ such that $|f(z)| \leq C |g(z)|$ for all $z > 0$.
			Using these bounds, we find that for $\gamma < 1$,
			$ \lim_{\lambda \to 0^+} c^2 \lambda^2 m'(-\lambda) + \tr\{ \Smat_{\epvec} \} \gamma( m(-\lambda) - \lambda m'(-\lambda) ) 
			= \lim_{\lambda \to 0^+}  c^2 O( \lambda^2 ) + \tr\{\Smat_{\epvec} \} \gamma ( (1 - \gamma)^{-1} + O(\lambda) - O(\lambda) ) 
			= \tr\{\Smat_{\epvec} \} \gamma (1 - \gamma)^{-1}$.
			Likewise, for $\gamma > 1$ we see that
			$ \lim_{\lambda \to 0^+} c^2 \lambda^2 m'(-\lambda) + \tr\{ \Smat_{\epvec} \} \gamma( m(-\lambda) - \lambda m'(-\lambda) ) 
			= \lim_{\lambda \to 0^+} c^2 \frac{\gamma - 1}{\gamma} + O(c^2 \lambda^2) + \tr\{\Smat_{\epvec} \} \gamma ( \frac{(\gamma - 1)}{\gamma \lambda} + \frac{1}{(\gamma - 1) \gamma} + O(\lambda) - \frac{(\gamma - 1)}{\gamma \lambda} - O(\lambda) )
			= c^2 \frac{\gamma - 1}{\gamma} + \tr\{\Smat_{\epvec} \} \frac{1}{\gamma - 1}$.
			Combining these expressions yields \eqref{eqn:limiting NIECE risk}, completing the proof.
		\end{proof}	
	
	}\fi
	
\end{document}